\documentclass[a4paper,USenglish,11pt]{article}
\usepackage[utf8]{inputenc}
\usepackage[margin=1in]{geometry}
\usepackage[shortlabels]{enumitem}
\usepackage{bm}
\usepackage{amsfonts}
\usepackage{amsmath}
\usepackage{xcolor}
\usepackage{graphicx}
\usepackage{amsthm}
\usepackage{amssymb}
\usepackage{qtree}
\usepackage{tree-dvips}
\usepackage{float}
\usepackage{authblk}
\usepackage{hyperref}
\usepackage[nameinlink]{cleveref}
\usepackage{thm-restate}
\newtheorem{theorem}{Theorem}
\newtheorem{lemma}{Lemma}
\newtheorem*{lemma*}{Lemma}
\newtheorem{corollary}{Corollary}
\newtheorem{proposition}{Proposition}

\newtheorem{definition}{Definition}

\newtheorem*{claim*}{Claim}

\theoremstyle{remark}
\newtheorem{remark}{Remark}
\newtheorem{example}{Example}

\usepackage[numbers,sort&compress]{natbib}

\DeclareMathOperator{\poly}{poly}
\DeclareMathOperator{\polylog}{polylog}
\DeclareMathOperator{\Char}{char}
\newcommand{\FF}{{\mathbb{F}}}
\newcommand{\QQ}{{\mathbb{Q}}}
\newcommand{\NN}{{\mathbb{N}}}
\newcommand{\Z}{{\mathcal{Z}}}
\newcommand{\Znew}{{\mathcal{Z}_\text{new}}}
\newtheorem{question}[theorem]{Question}

\title{Parameterized Sensitivity Oracles and Dynamic Algorithms using Exterior Algebras}

\author{Josh Alman\thanks{Department of Computer Science, Columbia University. Email: \href{mailto:josh@cs.columbia.edu}{josh@cs.columbia.edu}} , Dean Hirsch\thanks{Department of Computer Science, Columbia University. Email: \href{mailto:deanh@cs.columbia.edu}{deanh@cs.columbia.edu}}}

\date{April 2022}

\makeatletter
\providecommand*{\cupdot}{%
  \mathbin{%
    \mathpalette\@cupdot{}%
  }%
}
\newcommand*{\@cupdot}[2]{%
  \ooalign{%
    $\m@th#1\cup$\cr
    \sbox0{$#1\cup$}%
    \dimen@=\ht0 %
    \sbox0{$\m@th#1\cdot$}%
    \advance\dimen@ by -\ht0 %
    \dimen@=.5\dimen@
    \hidewidth\raise\dimen@\box0\hidewidth
  }%
}

\providecommand*{\bigcupdot}{%
  \mathop{%
    \vphantom{\bigcup}%
    \mathpalette\@bigcupdot{}%
  }%
}
\newcommand*{\@bigcupdot}[2]{%
  \ooalign{%
    $\m@th#1\bigcup$\cr
    \sbox0{$#1\bigcup$}%
    \dimen@=\ht0 %
    \advance\dimen@ by -\dp0 %
    \sbox0{\scalebox{2}{$\m@th#1\cdot$}}%
    \advance\dimen@ by -\ht0 %
    \dimen@=.5\dimen@
    \hidewidth\raise\dimen@\box0\hidewidth
  }%
}

\begin{document}

\maketitle

\begin{abstract}
We design the first efficient sensitivity oracles and dynamic algorithms for a variety of parameterized problems. Our main approach is to modify the algebraic coding technique from static parameterized algorithm design, which had not previously been used in a dynamic context. We particularly build off of the `extensor coding' method of Brand, Dell and Husfeldt [STOC'18], employing properties of the exterior algebra over different fields.

For the \textsc{$k$-Path} detection problem for directed graphs, it is known that no efficient dynamic algorithm exists (under popular assumptions from fine-grained complexity). We circumvent this by designing an efficient sensitivity oracle, which preprocesses a directed graph on $n$ vertices in $2^k \poly(k) n^{\omega+o(1)}$ time, such that, given $\ell$ updates (mixing edge insertions and deletions, and vertex deletions) to that input graph, it can decide in time $\ell^2 2^k\poly(k)$ and with high probability, whether the updated graph contains a path of length $k$. We also give a deterministic sensitivity oracle requiring $4^k \poly(k) n^{\omega+o(1)}$ preprocessing time and $\ell^2 2^{\omega k + o(k)}$ query time, and obtain a randomized sensitivity oracle for the task of approximately counting the number of $k$-paths. For \textsc{$k$-Path} detection in undirected graphs, we obtain a randomized sensitivity oracle with $O(1.66^k n^3)$ preprocessing time and $O(\ell^3 1.66^k)$ query time, and a better bound for undirected bipartite graphs.

In addition, we present the first fully dynamic algorithms for a variety of problems: \textsc{$k$-Partial Cover}, \textsc{$m$-Set $k$-Packing}, \textsc{$t$-Dominating Set}, \textsc{$d$-Dimensional $k$-Matching}, and \textsc{Exact $k$-Partial Cover}. For example, for \textsc{$k$-Partial Cover} we show a randomized dynamic algorithm with $2^k \poly(k)\polylog(n)$ update time, and a deterministic dynamic algorithm with $4^k \poly(k)\polylog(n)$ update time.
Finally, we show how our techniques can be adapted to deal with natural variants on these problems where additional constraints are imposed on the solutions.
\end{abstract}

	\section{Introduction}
	The area of dynamic algorithms studies how to quickly and efficiently solve computational problems when the input data is changing. For example, if $P$ is a property of a graph, then a dynamic graph algorithm for $P$ is a data structure which maintains an $n$-node graph $G$, and can handle updates which insert or remove an edge of $G$, and queries which ask whether $G$ currently satisfies $P$. Efficient dynamic algorithms, which handle updates and queries in $n^{o(1)}$ time, have been designed for many important problems, and they are used in many applications, both as ways to analyze evolving data, and as subroutines of larger algorithms which need to iterate over and check many similar possibilities. See, for instance, the recent survey~\cite{dynamic-algs-survey}.
	
	However, there are many prominent dynamic problems which would have many applications, but for which we do not have efficient algorithms. Often times, we even have conditional lower bounds from fine-grained complexity, showing that efficient dynamic algorithms for these problems are unlikely to exist (see e.g.~\cite{AW-lbs, HKNS-lbs, patrascu-lbs} and~\cite[{Section 2.1}]{dynamic-algs-survey}). It has recently become popular to circumvent such lower bounds by instead designing a \emph{sensitivity oracle} for the problem, a weaker notion which can still be used in many applications.
	
	Let $\ell$ be a positive integer. A sensitivity oracle for a dynamic problem, with sensitivity $\ell$, preprocesses an initial input, and must answer queries where $\leq \ell$ changes are made to the \emph{initial} input. For example, if $P$ is a property of a graph, then a graph algorithm for $P$ with sensitivity $\ell$ is a data structure which preprocesses an initial graph $G$, and can handle queries where $\ell$ edges are updated (inserted or removed) in the initial graph $G$, and asks whether $P$ is still satisfied. One can imagine `resetting' $G$ back to its original state after each query\footnote{Sensitivity oracles are sometimes referred to as `fault-tolerant' or `emergency planning' algorithms in the literature. For graph problems, these terms also sometimes refer to the decrement-only case (where edge updates only remove edges), but following~\cite[{Section A.1}]{so-hardness}, we use `sensitivity oracle' to refer to the fully dynamic case, where edges can be inserted and deleted.}.
	
	In this paper, we study dynamic algorithms and sensitivity oracles for \emph{parameterized} problems. Consider, for instance, the \textsc{$k$-Path} problem: given a positive integer $k$, in an $n$-node graph $G$ (directed or undirected), determine whether there is a path of length $k$. This problem is \textsf{NP}-complete, so we should not hope for a dynamic algorithm with update time $n^{o(1)}$ (such a dynamic algorithm could be used to solve the static problem in $n^{2 + o(1)}$ time!). However, \textsc{$k$-Path} is known to be fixed-parameter tractable (\textsf{FPT}), and can be solved in time $2^{O(k)} \cdot n^2$~\cite{narrow-sieves,williams-k-path,best-kpath-directed}, which is sufficiently efficient when $k$ is small. We can thus hope for dynamic parameterized algorithms for the problem, with update time $f(k) \cdot n^{o(1)}$. And indeed, a recent line of work has designed efficient dynamic parameterized algorithms with such a running time for many different problems, typically by using dynamic variants on classic techniques from the parameterized algorithms literature like kernelization and color coding. For the \textsc{$k$-Path} problem in \emph{undirected} graphs, such an algorithm is known with update time $k! \cdot 2^{O(k)} \cdot \polylog(n)$~\cite{dynamic-parameterized}, and another with amortized update time $2^{O(k^2)}$~\cite{dynamic-elim-forests}.
	
	By contrast, no efficient dynamic parameterized algorithm for \textsc{$k$-Path} in \emph{directed} graphs is known. Moreover, Alman, Mnich and Vassilevska~\cite{dynamic-parameterized} proved a conditional lower bound, that it does not have such an efficient dynamic parameterized algorithm assuming any one of three popular conjectures from fine-grained complexity theory (the \textsc{3SUM} conjecture, the \textsc{Triangle} conjecture, and a `\textsc{Layered Reachability Oracle}' conjecture they introduce, which concerns a special case of a more popular hypothesis about reachability oracles).
	
	This leads naturally to the two main questions we address in this paper. The first asks whether there is an analogue of the aforementioned line of work on sensitivity oracles for problems without efficient dynamic algorithms in the parameterized setting.
	
	\begin{question} \label{question:1}
	Is there an efficient parameterized sensitivity oracle for \textsc{$k$-Path} in directed graphs?
	\end{question}
	
	We say that a sensitivity oracle for a parameterized problem, with parameter $k$, is \emph{efficient} if its preprocessing time is $f(k) \cdot \poly(n)$ for some computable function $f$, and its query time is $\poly(\ell) \cdot g(k) \cdot n^{o(1)}$ for some computable function $g$ (where $\ell$ is the sensitivity parameter, i.e., the number of updates allowed per query). In the case of \textsc{$k$-Path} in directed graphs, and other graph problems, we specifically seek such an efficient sensitivity oracle in the \emph{fully dynamic} setting where queries can change any $\ell$ edges by inserting and deleting them.
	
	It is natural to ask that the query time has a \emph{polynomial} dependence on $\ell$ (rather than, say, just a $f(\ell)$ dependence), as we do here, for two reasons. First, this is the dependence one would get by converting an efficient dynamic algorithm into a sensitivity oracle. Second, with our definition, any parameterized problem with an efficient sensitivity oracle is in \textsf{FPT} via the algorithm where the sensitivity oracle preprocesses an empty graph and then gets the full input graph as a query (but this would not be true if an arbitrary $f(\ell)$ term were allowed in the query time). Along the way to answering Question~\ref{question:1}, we will also address more precisely the relationship between the classes of parameterized problems with efficient dynamic algorithms, efficient sensitivity oracles, and efficient static algorithms (a.k.a. the class \textsf{FPT}). 
	
	To our knowledge, such a \emph{fully dynamic} notion of sensitivity oracles for parameterized problems has not been previously studied. The closest prior work is very recent~\cite{so-decremental} which considered a similar but \emph{only decremental} setting, wherein queries may only delete edges from the graph, and not insert new edges. They design very elegant decremental sensitivity oracles for directed \textsc{$k$-Path} and for $k$-Vertex Cover, but their preprocessing and query times have exponential dependence on $\ell$ and hence are not `efficient' as we defined above. We also give evidence that the techniques of~\cite{so-decremental} cannot extend to the fully-dynamic setting; see \Cref{sec:relatedwork} below for more details.
	
	The second question we address is inspired by prior work on static algorithms for \textsc{$k$-Path}. Many fundamental techniques in the literature on parameterized algorithms were first introduced to study the \textsc{$k$-Path} problem. One such technique, algebraic coding (sometimes called `monomial testing' or `multilinear monomial detection'), is used in the current fastest static randomized algorithms for \textsc{$k$-Path}, and has also been used in other applications in algebraic complexity theory \cite{algebraic-problems, best-kpath-directed, extensor-coding, patching-colors, narrow-sieves, koutis-constraints}. Nonetheless, to our knowledge, these techniques have not been used in a dynamic or sensitivity setting before.
	
	\begin{question} \label{question:2}
	Can algebraic coding techniques from the design of parameterized algorithms be used to design efficient dynamic algorithms or sensitivity oracles?
	\end{question}
	
	A positive answer to Question~\ref{question:2} could lead to efficient dynamic algorithms or sensitivity oracles for a host of parameterized problems.
	
	\subsection{Our results}
	
	Let $\omega < 2.373$ be such that we can multiply two $n \times n$ matrices in $O(n^\omega)$ arithmetic operations~\cite{matrix-mult-exponent}. Our first main result gives a positive answer to Question~\ref{question:1}.
	\begin{theorem} \label{thm:main1}
	    The \textsc{$k$-Path} problem in directed graphs has an efficient parameterized sensitivity oracle. It can be solved with\footnote{We work in the word-RAM model of computation with $w$-bit words for $w = O(\log n)$. Hence, only $O(\ell)$ words are needed to specify the $\ell$ edges to change in a query, and we can achieve query times independent of $n$.}:
	    \begin{itemize}
	        \item a Monte Carlo randomized algorithm with preprocessing time $2^k \poly(k) n^\omega$ and query time $\ell^2 2^k \poly(k)$, or
	        \item a deterministic algorithm with preprocessing time $4^k \poly(k) n^\omega$ and query time $\ell^2 2^{\omega k}$.
	    \end{itemize}
	    In addition to edge insertion and deletions, these algorithms also allow for vertex failures as part of the $\ell$ updates per query.
	\end{theorem}
	
	Although \textsc{$k$-Path} is known to not have an efficient dynamic parameterized algorithm (assuming the aforementioned \textsc{3SUM}, \textsc{Triangle}, or \textsc{Layered Reachability Oracle} hardness assumptions from fine-grained complexity), \Cref{thm:main1} shows it does have an efficient parameterized sensitivity oracle. Since the reductions used in \cite{dynamic-parameterized} are still valid in the sensitivity setting, one corollary is that the sensitivity versions of the \textsc{3SUM}, \textsc{Triangle}, and \textsc{Layered Reachability} problems have efficient algorithms (although there are simple algorithms showing this for \textsc{3SUM} and \textsc{Triangle} which do not go through \textsc{$k$-Path}; see \Cref{sec:dynamic-vs-sensitivity} for more details). 
	
	It follows from prior work~\cite{so-hardness} that assuming another popular conjecture, the Strong Exponential Time Hypothesis (\textsf{SETH}), there are problems in \textsf{FPT} which do not have efficient parameterized sensitivity oracles (one example is the counting version of the single-source reachability problem; see again \Cref{sec:dynamic-vs-sensitivity} for more details).
	Hence, assuming popular conjectures from fine-grained complexity, it follows that the class of parameterized problems with an efficient sensitivity oracle lies \emph{strictly between} the class of parameterized problems with an efficient dynamic algorithm, and the class \textsf{FPT} (i.e., both class inclusions are strict).
	
	Our sensitivity oracle for \Cref{thm:main1} uses $\Theta(n^2)$ space (for constant $k$), and it is natural to wonder whether this can be improved, especially since known dynamic parameterized algorithms for many other problems use much smaller space (e.g., many which dynamically maintain a small kernel~\cite{dynamic-parameterized}). However, we prove unconditionally that the space usage of our algorithm cannot be improved:
	\begin{theorem} \label{thm:main-space}
	    Any randomized or deterministic sensitivity oracle for \textsc{$k$-Path} in directed graphs which handles edge insertion queries must use $\Omega(n^2)$ space.
	\end{theorem}
	
	We also give three additional algorithmic results to complement \Cref{thm:main1}. First, we extend \Cref{thm:main1} to give an algorithm for approximately counting $k$-paths:
	\begin{theorem} \label{thm:main-approx-counting}
	    There is a randomized efficient parameterized sensitivity oracle which approximately counts the number of $k$-paths in an $n$-node, $m$-edge directed graph. For any $\epsilon > 0$, it produces an estimate to the number of $k$-paths in the graph that, with probability $> 99 \%$, is within $\epsilon$ relative error, with preprocessing time $\epsilon^{-2} \cdot 4^k \poly(k) \cdot \min\{ mn, n^\omega \}$ and update time $O\left(\epsilon^{-2} \cdot \ell^2 \cdot 2^{\omega k}\right)$.
	    In addition to edge insertion and deletions, it also allows for vertex failures as part of the $\ell$ updates per query.
	\end{theorem}
	
	Second, we present a randomized sensitivity oracle with a better dependence on $k$, at the cost of a worse dependence on $\ell$, but only for undirected graphs:
	
\begin{restatable}[Undirected graphs]{theorem}{undirectedfastso} \label{thm:so-undirected-fast}
	For the \textsc{$k$-Path} detection problem on an undirected graph $G$ on $n$ vertices, there exists a randomized sensitivity oracle with preprocessing time $O(1.66^k n^3)$ and query time $O(\ell^3 1.66^k)$.
\end{restatable}

Third, we obtain the following corollary for bipartite graphs:

\begin{restatable}[Undirected bipartite  graphs]{corollary}{undirectedbipartiteso} \label{cor:undirected-bipartite}
	For the \textsc{$k$-Path} detection problem on an undirected bipartite graph, where the partition of the vertices $V$ into the two sides $V=S \cup T$ is known in advance and holds after any updates, there exists a sensitivity oracle with $2^{k/2}\poly(k) n^3$ preprocessing time and $\ell^3 2^{k/2} \poly(k)$ query time.
\end{restatable}
	
	Our algorithm for \Cref{thm:main1} uses a new dynamic version of the algebraic coding technique, and more specifically, a recent implementation called `extensor coding' by Brand, Dell, and Husfeldt~\cite{extensor-coding}. See \Cref{sec:extensor-coding} and \ref{sec:overview} below for more background on this technique, and the new ideas we introduce to be able to use it in a dynamic or sensitivity setting. We are then able to use and further modify our approach to design efficient dynamic parameterized algorithms for many problems for which no such algorithm was previously known, positively answering Question~\ref{question:2}. We present the definitions for the problems we consider, followed by the results.

	\begin{restatable}[\textsc{$k$-Partial Cover}]{definition}{partialcover} \label{def:partial-cover}
		Given a collection of subsets $S_1,...,S_n\subseteq [N]$, find the minimum size of a sub-collection $T$ of these, for which $\big|\bigcup_{S\in T} S\big| \ge k$, or declare that no such $T$ exists.
	\end{restatable}
	\begin{restatable}[\textsc{$m$-Set $k$-Packing}]{definition}{setpacking} \label{def:m-set-k-packing}
	Given subsets $S_1,...,S_n\subseteq [N]$, all with size $|S_i|=m$, decide whether there exists a sub-collection of $k$ sets, which are all pairwise disjoint.
	\end{restatable}
	\begin{restatable}[\textsc{$t$-Dominating Set}]{definition}{tdominating} \label{def:t-dominating}
		Given an undirected graph $G$ on $n$ vertices, find the minimum size of a set of vertices $S\subseteq V(G)$ such that $|S\cup N(S)|\ge t$ where $N(S)$ is the set of neighbors of vertices in $S$, i.e., $N(S)=\bigcup_{v\in S} N(v)$.
	\end{restatable}
	\begin{restatable}[\textsc{$d$-Dimensional $k$-Matching}]{definition}{dimensionalmatching} \label{def:d-dim-k-matching}
		Fixing a universe $U=U_1\times U_2\times \ldots \times U_d$, where $U_i$ are pairwise disjoint of combined size $|\bigcupdot_i U_i|=N$, and given a collection of tuples $T\subseteq U$, decide whether $T$ contains a sub-collection of $k$ pairwise disjoint tuples. (We say the tuples $a,b \in U$ are disjoint if $a[i] \neq b[i]$ for all $i = 1,2, \ldots,d$.)
	\end{restatable}
	\begin{restatable}[\textsc{Exact $k$-Partial Cover}]{definition}{exactpartialcover} \label{def:exact-partial-cover}
		Given a collection of subsets $S_1,...,S_n\subseteq [N]$, decide whether there is a sub-collection $T$ of these, in which all sets are pairwise-disjoint, and $\big|\bigcupdot_{S\in T} S\big| = k$.
	\end{restatable}
	\begin{theorem}
	There are efficient dynamic parameterized algorithms for the following problems.
	We write $O^*$ here to hide factors that are polynomial in the parameters ($m, k, t$) and polylogarithmic in the size of the instance ($n$, $N$).
    \begin{itemize}
        \item 
    	For \textsc{$k$-Partial Cover} there are dynamic algorithms, with either randomized $O^*(2^k)$ or deterministic $O^*(4^k)$ update time (see \Cref{sec:partial-cover}).
	    \item
    	For \textsc{$m$-Set $k$-Packing} there are dynamic algorithms, with either randomized $O^*(2^{mk})$ or deterministic $O^*(4^{mk})$ update time (see \Cref{sec:set-packing}).
    	\item
		For \textsc{$t$-Dominating Set} there are dynamic algorithms, with either randomized $O^*(2^t)$ or deterministic $O^*(4^t)$ update time (see \Cref{sec:t-dominating}).
		\item
		For \textsc{$d$-Dimensional $k$-Matching} there are dynamic algorithms, with either randomized $O^*(2^{(d-1)k})$ or deterministic $O^*(4^{(d-1)k})$ update time (see \Cref{sec:dimensionalmatching}).
		\item
		For \textsc{Exact $k$-Partial Cover} there are dynamic algorithms, with either randomized $O^*(2^k)$ or deterministic $O^*(4^k)$ update time (see \Cref{sec:exact-partial-cover}).
	\end{itemize}
	\end{theorem}
	
	We also explain how the results can be extended to problems with additional constraints. As a simple example, we show how we can design an efficient sensitivity oracle for the problem of detecting whether a directed graph contains a walk of length $k$ which visits at least $k-1$ vertices.
    As a second example, we discuss a different problem: given a directed graph $G$ on $n$ vertices, two (possibly intersecting) subsets $V_1,V_2\subseteq V$, and two positive integers $\mu_1, \mu_2$, decide whether $G$ contains a $k$-path that contains at most $\mu_1$ vertices of $V_1$ and at most $\mu_2$ vertices of $V_2$. For this problem we give a static deterministic algorithm running in time $4^{k+\min\{k,|V_1\cap V_2|\}} \poly(k) \cdot n^2$, as well as a sensitivity oracle counterpart.
	We refer the reader to \Cref{sec:constraints} for the full discussion.
	
	\subsection{Comparison with related work} \label{sec:relatedwork}
	
	\paragraph*{Algorithms for \textsc{$k$-Path}}
	The static \textsc{$k$-Path} problem has a long history. The current best upper bounds for it are a randomized $1.66^k \poly(n)$ algorithm in undirected graphs due to Björklund, Husfeldt, Kaski and Koivisto~\cite{narrow-sieves}, a randomized time $2^k\poly(n)$ in directed (and hence also undirected) graphs due to Williams~\cite{williams-k-path}, and $2.554^k \poly(n)$ deterministic time in directed graphs due to Tsur~\cite{best-kpath-directed}.
	
	In comparison, our randomized sensitivity oracle for directed graphs in \Cref{thm:main1} has a dependency on $k$ (both in preprocessing and query time) of $O^*(2^k)$, so one cannot hope to improve on it without improving the best static algorithm. For undirected graphs, we give in \Cref{thm:so-undirected-fast} a sensitivity oracle with a dependency of $O^*(1.66^k)$, which matches the best bound of \cite{narrow-sieves}. Our deterministic sensitivity oracle has a dependency of $O^*(4^k)$; we leave open the question of whether the techniques of~\cite{best-kpath-directed} can be used in a sensitivity oracle setting to achieve $O^*(2.554^k)$.
	
	There has also been work on the dynamic version of \textsc{$k$-Path}. However, Alman, Mnich and Vassilevska~\cite{dynamic-parameterized} showed that under the aforementioned fine-grained conjectures, there is no dynamic algorithm for \emph{directed} graphs. At the same time, they give a deterministic dynamic algorithm for the case of \emph{undirected} graphs, with update time $k! \cdot 2^{O(k)} \cdot \polylog(n)$. A subsequent result due to Chen et al.~\cite{dynamic-elim-forests} presents a deterministic dynamic algorithm with amortized update time $2^{O(k^2)}$.
	
	In comparison, our sensitivity oracle works for directed graphs just as well as for undirected graphs, and with a lower dependence on $k$ (only $2^{O(k)}$). However, we obtain a sensitivity oracle, rather than a dynamic algorithm.
	
	\paragraph*{Decremental Parameterized Sensitivity Oracles}
	
	A recent paper by Bilò et al.~\cite{so-decremental} constructs decremental sensitivity oracles (``fault tolerant'') for the \textsc{$k$-Path} problem (i.e., where all updates are edge deletions). They give one construction with a randomized $k^\ell 2^k \poly(n)$ preprocessing time and $O(\ell \min\{\ell, k+\log \ell\})$ update time, and a second construction with a lower preprocessing time of $2^k\cdot\frac{(\ell+k)^{\ell+k}}{\ell^\ell k^k}\cdot\ell \poly(n)$ at the expense of an increased query time of $O\left(\frac{(\ell+k)^{\ell+k}}{\ell^\ell k^k}\cdot \ell \min\{\ell, k\}\log n \right)$. They also give a deterministic algorithm, requiring $k^\ell 2.554^k \poly(n)$ preprocessing time and a very low $O(\ell \min\{\ell, k+\log \ell\})$ update time. Our \Cref{thm:main1} improves on the dependency on $\ell$, making it polynomial in both the preprocessing and update times, and allows for both increments and decrements, while the methods of \cite{so-decremental} are specific to decrements.
	We also achieve an update time which is independent of $n$ in the word-RAM model. 
	In \cite{so-decremental}, preprocessing is made aware of the value of $\ell$, due to the superpolynomial dependence on it.
	
	The techniques of Bilò et al.~\cite{so-decremental} also achieve a very low memory footprint for their fault-tolerant oracle, attaining at most logarithmic space dependency on $n$ in all variants, when $k$ and $\ell$ are considered constants. In contrast, we prove a lower bound, showing that one cannot hope for such dependency if increments are allowed, and that any fully dynamic sensitivity oracle must store at least $\Omega(n^2)$ bits (see \Cref{thm:main-space}). This suggests that the methods used in \cite{so-decremental} cannot be used to devise a fully dynamic sensitivity oracle. Their algorithms are based on simple and elegant combinatorial arguments; for instance, in their algorithm with a faster preprocessing time, they precompute a collection of $k$-paths in the graph, in a way that ensures that with high probability, if there exists a path after $\ell$ edge deletions, one of the precomputed paths is still valid. We instead take a very different approach based on algebraic coding.
	
	\paragraph*{Existing distance sensitivity oracles}
	It has come to our attention that results of van den Brand and Saranurak~\cite{algebraic-distance-oracles} on distance sensitivity oracles, combined with standard color-coding techniques, imply the existence of an efficient sensitivity oracle for directed \textsc{$k$-Path} with $k^{O(k)}n^\omega$ preprocessing time and $\ell^\omega k^{O(k)}$ query time, positively answering our \Cref{question:1}. The sensitivity oracles presented in this paper for this problem (\Cref{thm:kpath-so-rand} and \Cref{thm:kpath-so-det}) achieve better dependence on the parameters $k$ and $\ell$.
	
	\paragraph*{Cover, Matching, Dominating Set, and Packing problems}
	We are unaware of previous work on dynamic algorithms for the problems we study here (other than $k$-Path in undirected graphs, which we discussed above). Instead, we compare the dependence on $k$ in the update times of our dynamic running times with the best known dependence on $k$ for static algorithms. In the cases where we match, it is impossible to speed up the dependence on $k$ in our dynamic algorithms without speeding up the fastest known static algorithms.
	\begin{enumerate}
    \item For both \textsc{$k$-Partial Cover} and \textsc{$t$-Dominating Set}, the fastest known static randomized running time is $2^k \poly(nk)$ by Koutis and Williams~\cite{algebraic-problems} (where for \textsc{$t$-Dominating Set} we replace $k$ with $t$ in the running time), which our dynamic algorithm matches, and the fastest known static deterministic running time is $2.554^{k}\poly(kn)$, stated by Tsur~\cite{best-kpath-directed}, whereas our deterministic dynamic algorithm achieves query time $O^*(4^k)$.
    
    \item The fastest known static randomized algorithm for \textsc{$m$-Set $k$-Packing} by Björklund, Husfeldt, Kaski and Koivisto~\cite{narrow-sieves} has the intimidating running time of
    \[ \left(\frac{0.108157\cdot 2^m(1-1.64074/m)^{1.64076-m}m^{0.679625}}{(m-1)^{0.679623}} \right)^k n^6 \poly(N) \]
    
    This is less than $2^{mk}\poly(n)$, and considerably so for smaller $m$. For example, when $m=3$, the running time is bounded by $1.4953^{3k}n^6 \poly(N)$.
    Due to the super-quadratic dependence on $n$, these techniques are unlikely to be adaptable for the dynamic case.
    
    In contrast, Koutis~\cite{koutis-packing} proposes a static randomized algorithm running in time \[ 2^{mk}n \poly(mk) \polylog(n).\]
    Our dynamic result matches this running time.
    \item The fastest known static randomized algorithm for \textsc{$d$-Dimensional $k$-Matching}, by Björklund, Husfeldt, Kaski, and Koivisto~\cite{narrow-sieves}, runs in time $2^{(d-2)k} \poly(Nk) n$. In contrast, we match only an earlier algorithm by Koutis and Williams~\cite{algebraic-problems}, which runs in time $2^{(d-1)k}\poly(kdn)$. The fastest known static deterministic algorithm runs in time $2.554^{(d-1)k}\poly(Nn)$, stated in Tsur~\cite{best-kpath-directed}, whereas we achieve deterministic dynamic update time $O^*(4^{(d-1)k})$.
    \item \textsc{Exact $k$-Partial Cover} is a natural generalization of similar problems (such as \textsc{$m$-Set $k$-Packing}), though we are unaware of explicit prior work on it. A slight modification of the algorithm of Koutis~\cite{koutis-packing} for \textsc{$m$-Set $k$-Packing} solves the static problem in randomized $O^*(2^k)$ time, matching the dependence on $k$ in our dynamic algorithm.
	\end{enumerate}
	
	\paragraph*{Dynamic Parameterized Algorithms}
	Efficient dynamic algorithms have previously been proposed for a number of parameterized algorithms, often using dynamic versions of classic techniques from the design of fixed-parameter tractable algorithms.
	These techniques include dynamic kernels \cite{dynamic-parameterized, dynamic-kernels, parameterized-streaming, Iwata2014FastDG}, color-coding \cite{dynamic-parameterized}, dynamic elimination forests \cite{dynamic-elim-forests}, dynamic branching trees \cite{dynamic-parameterized}, sketching \cite{parameterized-streaming}, and other methods \cite{dynamic-tree-depth-decomposition, dynamic-parameterized-counting}.
	To our knowledge, no previous work on dynamic parameterized algorithms has used algebraic techniques.
    
    \subsection{Subsequent algorithm for skew-multliplication}\label{sec:subsequent-skew-prod}
    Shortly after the publication of this paper, Brand~\cite{fast-deterministic-skew-product} presented an improvement to a primitive we are consistently using throughout this paper when devising the deterministic variants of our algorithms. Specifically, skew-multiplication in the sub-algebra we are using for deterministic algorithms can be done more efficiently than assumed here, replacing every occurrence of $4^k$ in our running times with $\varphi^{2k}\poly(k)$, where $\varphi=\frac{1+\sqrt{5}}{2}< 1.619$. This applies to all our deterministic algorithms, presented in Theorems~\ref{thm:kpath-so-det},\ref{thm:approx-count},\ref{thm:exact-set-cover},\ref{thm:partial-cover},\ref{thm:m-set-k-packing},\ref{thm:t-dominating}, and \ref{thm:k-matching}. We emphasize that this does not improve the skew-multiplication in the most general form that appears in \Cref{par:skew_prod}, but rather improves specifically the operation for the structure used in the deterministic variants. For more precise discussion see \Cref{par:skew-prod-det}.
    
	\section{Preliminaries}	
	
	\subsection{Notation}
	For a positive integer $k$, write $[k] := \{1,2,\ldots,k\}$.
	
	We will use the asymptotic notation $O^*(T)$; its meaning will depend on context, so we will define it each time it arises. Generally, in static algorithms and preprocessing times we hide factors that are polynomial in the input parameters ($k$, $n$, etc), but when talking about update times, we will only hide factors that are polylogarithmic in the input size.
	
	Let $\omega$ be such that two $n \times n$ matrices can be multiplied using $O(n^{\omega})$ field operations\footnote{This is a slight abuse of notation. The exponent of matrix multiplication $\omega$ is typically defined as the smallest constant such that $n \times n$ matrices can be multiplied in $n^{\omega + o(1)}$ field operations, but we drop the `$o(1)$' in the exponent here for simplicity.}; it is known that we can take $\omega < 2.373$~\cite{matrix-mult-exponent}.
	\subsection{Sensitivity Oracles}
	
	\begin{definition}[sensitivity oracle] \label{def:sensitivity-oracle}
	A \emph{sensitivity oracle} is a data structure containing two functions:
	\begin{enumerate}
	    \item Initialization with an input instance of size $n$.
	    \item Query with $\ell$ changes (with problem-specific definition) to the \emph{original input}. This returns the desired output on the altered input.
	\end{enumerate}
	The time spent in the initialization step is called the \emph{preprocessing time}, and the time spent in the query step is called the \emph{query time}.
	\end{definition}
	
	In random sensitivity oracles, we assume the queries are independent of the randomness of the oracle, and thus cannot be used to fool a sensitivity oracle by obtaining the randomness used at the preprocessing phase. We say that a random sensitivity oracle errs with probability at most $p$, if for any pair of an initial configuration and a query, the probability of error is at most $p$.
	
	We say that a sensitivity oracle for a parameterized problem, with parameter $k$, is \emph{efficient} if its preprocessing time is $f(k)\poly(n)$ for some computable function $f$, and its query time is $\poly(\ell) g(k)n^{o(1)}$ for some computable function $g$.
	
	\subsection{Exterior algebra}
	
	A key technical tool we will make use of in our algorithms is the exterior algebra. We give a brief introduction to its definition and properties which we will use.
	
	Given a field $\FF$ and a positive integer $k$, the exterior algebra $\Lambda(\mathbb{F}^k)$ is a non-commutative ring over $\FF$. In this paper we will be interested in the cases where $\FF$ is either the rational numbers $\QQ$, or a finite field $\FF_{2^d}$ of characteristic 2.
	
	Following the terminology of
	Brand, Dell and Husfeldt \cite{extensor-coding}, elements of the exterior algebra are called \emph{extensors}.
	Multiplication in $\Lambda(\FF^k)$ is denoted by the wedge sign; for $x,y\in\Lambda(\FF^k)$, their product is $x\wedge y$.
	The generators of this ring are $e_1,e_2,...,e_k$, the standard basis of $\FF^k$, and we impose the relations $e_i\wedge e_j = -e_j \wedge e_i$ for all $i,j\in [k]$ and $e_i\wedge e_i=0$ for all $i\in [k]$ (while the latter equation follows from the former for fields of characteristic different from 2, we still require it when $\Char \FF=2$). We also require that the wedge product is bilinear\footnote{$(a+b)\wedge c = (a \wedge c) + (b \wedge c)$ and $a \wedge (b+c) = (a \wedge b) + (a \wedge c)$.}. Furthermore, it can be seen that the wedge product is associative\footnote{$(a \wedge b) \wedge c = a \wedge (b \wedge c)$}.
	
	Additionally, we consider any field element $a\in \FF$ to be an element of $\Lambda(\FF^k)$, with multiplication defined simply by $a\wedge e_i = e_i \wedge a = ae_i$.
	
	For example, we have
	\begin{align*}
	    e_1\wedge (e_5 - e_2 + 3)\wedge e_3 &= e_1\wedge e_5 \wedge e_3 - e_1 \wedge e_2 \wedge e_3 + 3 e_1 \wedge e_3 \\ & = - e_1\wedge e_3 \wedge e_5 - e_1 \wedge e_2 \wedge e_3 + 3 e_1 \wedge e_3.
	\end{align*}
	
	To simplify the notation, for $I\subseteq [k]$ we also write $e_{I}:=\bigwedge_{i\in I} e_i$, where the product is over the elements of $I$ in ascending order, with the convention $e_\emptyset=1$. We will sometimes also replace the wedge sign by a simple product sign, when it is clear from context that we mean the wedge product. For example, we can write $e_{I}:=\prod_{i\in I} e_i$. Thus, the above example extensor can also be written as 
	\[ e_1\wedge (e_5 - e_2 + 3)\wedge e_3 = - e_{\{1,3,5\}} - e_{\{1,2,3\}} + 3 e_{\{1,3\}} \]
	
	We see that two products of the basis elements that differ only by the order of the terms can differ only in sign (e.g., $e_1 \wedge e_2 \wedge e_3 = - e_2 \wedge e_1 \wedge e_3$), and furthermore any product of $e_i$s with a repeating factor vanishes (e.g., $e_1 \wedge e_2 \wedge e_2 = e_1 \wedge 0 = 0$). Thus, the set $\{e_I:I\subseteq [k]\}$ spans $\Lambda(\FF^k)$ as a vector space over $\FF$. It is in fact a basis of this vector space, and we have $\dim_{\FF} \Lambda(\FF^k) = 2^k$. That is, for any $x\in\Lambda(\FF^k)$ there is a unique choice of the $2^k$ coefficients $\alpha_I$ so that $x=\sum_{I\subseteq [k]} \alpha_I e_I$. For any $x\in\Lambda(\FF^k)$ and $T\subseteq [k]$ we denote by $[e_T]x$ the coefficient of $e_T$ when $x$ is represented in the basis $\{e_I:I\subseteq [k]\}$ (i.e., the value of $\alpha_T$ as above).
	
	If there is a $d$ such that all $I$ for which $[e_I]x\neq 0$ have $|I|=d$, we say that $x$ is a \emph{degree-$d$} extensor. The space of degree-$d$ extensors is denoted by $\Lambda^d(\FF^k)$.
	
	We identify any vector $v=(v[1],v[2],...,v[k]) \in \FF^k$ with the exterior algebra element $\sum_i v[i] e_i \in \Lambda^1(\FF^k)$. A crucial property of the exterior algebra is that for any vector $v\in \FF^k$ it holds that $v\wedge v = 0$. Indeed,
	\begin{align*} v\wedge v = \sum_{i,j} v[i] v[j] e_i \wedge e_j &= \sum_{i<j} v[i] v[j] e_i \wedge e_j + \sum_{i<j} v[i] v[j] e_j \wedge e_i \\
	&= \sum_{i<j} v[i] v[j] e_i \wedge e_j - \sum_{i<j} v[i] v[j] e_i \wedge e_j = 0.
	\end{align*}
	Similarly, for any two vectors $u,v\in\FF^k$, it holds that $u\wedge v = -v\wedge u$. It is important to notice that these properties do not hold for all elements in $\Lambda(\FF^k)$. As an example, for $x=e_1+e_2\wedge e_3$ we have $x\wedge x = (e_1+e_2\wedge e_3)\wedge (e_1+e_2\wedge e_3) = e_1\wedge e_2\wedge e_3 + e_2\wedge e_3\wedge e_1 = 2 e_1\wedge e_2\wedge e_3$. In fact, we can see that the wedge product of an even number of vectors commutes with any extensor.
	
	Another very useful property of the exterior algebra is its connection to determinants. Let $v_1,v_2,...,v_k\in\FF^k$ be $k$ vectors of dimension $k$, and consider their product $v_1\wedge v_2\wedge ... \wedge v_k$. Replace each $v_i$ with its linear combination of the basis elements $e_j$ and expand the product. We see that any monomial that repeats a basis element gets cancelled, and the others are all $e_{[k]}$ up to a sign. We get
	$$v_1\wedge v_2\wedge ... \wedge v_k = \bigwedge_{i=1}^k \sum_{j=1}^k v_i[j] e_j = \sum_{\sigma\in S_k} \bigwedge_{i=1}^k v_i[\sigma(i)] e_{\sigma(i)} = \sum_{\sigma\in S_k} \operatorname{sgn}(\sigma)e_{[k]} \bigwedge_{i=1}^k v_i[\sigma(i)].$$
	We recognize a determinant in the right hand side, thus showing that
	\begin{equation} \label{eq:determinant}
	    v_1\wedge v_2\wedge ... \wedge v_k = \det(v_1|v_2|...|v_k) e_{[k]}.
	\end{equation}
	
	\subsection{Complexity of ring operations in the exterior algebra}
	We represent the elements of $\Lambda(\FF^k)$ as the length-$2^k$ vector of coefficients of the basis elements $e_I$.
	Addition is performed by coordinate-wise addition, hence requires $O(2^k)$ field operations.
	Multiplication is trickier, and there are a few important cases. We have the following main cases, which were also noted and used in \cite{extensor-coding}.
	
	\begin{proposition}[Skew product]\label{par:skew_prod}
	Computing $x\wedge v$ for a general extensor $x\in\Lambda(\FF^k)$ and a vector $v\in\FF^k$ can be done in $O(2^k\cdot k)$ field operations.
	\end{proposition}
	\begin{proof}[Proof sketch]
	This is accomplished by simply expanding the product.
	\end{proof}
	
	\begin{proposition}[General product over characteristic 2] \label{par:general-prod-char-2}
	Computing $x\wedge y$ for any $x,y\in\Lambda(\FF^k)$, when $\Char (\FF)=2$, can be done in $O(2^k k^2)$ field operations.
	\end{proposition}
	\begin{proof}[Proof sketch]
	Here, the ring $\Lambda(\FF^k)$ is commutative. Then we can write
	\[[e_I](x\wedge y) = \sum_{T\subseteq I} ([e_T]x)([e_{I\setminus T}]y),\]
	which can be seen as a subset convolution operation. Using the fast subset convolution discovered by Björklund, Husfeldt, Kaski and Koivisto \cite{fast-subset-convolution}, we can compute this with $O(2^k k^2)$ field operations.
	\end{proof}

    \begin{proposition}[General product over any characteristic] \label{par:general-prod}
	Computing $x\wedge y$ for $x,y\in\Lambda(\FF^k)$, for $\FF$ of any characteristic, can be done in $O(2^{\omega k / 2})$ field operations.
    \end{proposition}
    \Cref{par:general-prod} is a result of Włodarczyk \cite{extensor-mult}, which reduces computing $x\wedge y$ over $\Lambda(\FF^k)$ to $k^2$ multiplications in a Clifford algebra, which in turn can be embedded in matrices of size $2^{k/2}\times 2^{k/2}$.
    
    The following is a considerably more specialized case due to Brand~\cite{fast-deterministic-skew-product} that first appeared subsequent to this paper, but exponentially improves our results for the deterministic algorithms, as mentioned in \Cref{sec:subsequent-skew-prod}. It concerns a specific subalgebra of $\Lambda(\FF^{2k})$, namely the ring generated by the unit $e_\emptyset$ and the elements of the form $\binom{v}{0}\wedge \binom{0}{v}$, where $v\in\FF^k$ is an arbitrary vector and $0$ is the zero vector $0\in \FF^k$.
    Let us denote this ring by $D(\FF^{2k})$.
    
    \begin{theorem}[Skew product in $D(\QQ^{2k})$, \cite{fast-deterministic-skew-product}] \label{par:skew-prod-det}
	We can represent elements in $D(\QQ^{2k})$ with $\varphi^{2k}\poly(k)$ field values, such that computing $x\wedge y$ for $x,y\in D(\QQ^{2k})$ with $y=\binom{v}{0}\wedge \binom{0}{v}$ for some $v\in \QQ^k$, can be done in $\varphi^{2k}\poly(k)$ field operations. Adding two elements can be done with the same time bound.
    \end{theorem}
	
	\subsection{Extensor-coding} \label{sec:extensor-coding}
	In this subsection, we give a brief overview of the techniques recently used by Brand, Dell and Husfeldt~\cite{extensor-coding} to design a randomized and a deterministic algorithm for the (static) \textsc{$k$-Path} problem; we heavily build off of these techniques in this paper.
	
	Given a directed graph $G$, we denote its vertex set $V=\{v_1,...,v_n\}$, and we denote by $W_s(G)$ the set of all walks in $G$ of length $s$. Recall that walks may repeat vertices, whereas paths may not. For each edge $e\in E(G)$ we have an \emph{edge variable} $y_e$, whose possible values will be in a field $\FF$ we will pick later.
	
	Furthermore, we define vectors $\chi:V\to \FF^k$ by $\chi(v_i)=(1, j, j^2, \ldots, j^{k-1})$ where $j = f(i)$ for some injective function $f : [n] \to \FF$. We call these \emph{Vandermonde vectors}.
	
	Given a walk in $G$ of length $s$, $w=(w_1,w_2,...,w_s)\in W_s(G)$, we define the corresponding \emph{walk extensor} to be $\chi(w_1) \wedge y_{w_1w_2} \wedge \chi(w_2) \wedge y_{w_2w_3} \wedge \chi(w_3) \wedge ... \wedge y_{w_{s-1}w_s} \wedge \chi(w_s)$. We will sometimes denote the walk extensor of a walk $w$ by $\chi(w)$.
	
	Our goal is to compute the sum of all walk extensors for walks of length $k$,
	
    \begin{equation} \label{eq:walk-extensor-sum}
        \Z = \sum_{(w_1,w_2,...,w_k)\in W_k(G)} \chi(w_1) \wedge y_{w_1w_2} \wedge \chi(w_2) \wedge y_{w_2w_3} \wedge \chi(w_3) \wedge ... \wedge y_{w_{k-1}w_k} \wedge \chi(w_k).
    \end{equation}
	The result of this sum is seen to be proportional to the single basis element $e_{[k]}$, and by abuse of notation we will consider it as a scalar equal to that coefficient (or, more precisely, a polynomial in the $y$ variables).
	
	Any $k$-walk that is not a path does not contribute to \Cref{eq:walk-extensor-sum}, because its walk extensor repeats a vector in the product. Thus, any nonzero term in the sum \Cref{eq:walk-extensor-sum} corresponds to a $k$-\emph{path}.
	
	The choice of $\chi(v_i)$ is such that for any walk $(w_1,w_2,...,w_k)$ that is a $k$-path, the walk-extensor is a nonzero monomial in the $y$ variables. This is seen using \Cref{eq:determinant}, because $\chi(w_1)\wedge \chi(w_2) \wedge ... \wedge \chi(w_k) = e_{[k]} \det(\chi(w_1)|\chi(w_2)|...|\chi(w_k))$, which is the determinant of a Vandermonde matrix with unique columns, which is known to be nonzero.
	
	Additionally, the monomial in the $y$ variables given to each $k$-path is seen to uniquely identify the $k$-path (since there is a different $y$ variable for each edge). Hence, the monomials of different $k$-paths cannot cancel, and we have proven:
	
	\begin{lemma} \label{lem:k-path-existence}
	For a directed graph $G$, the value of $\Z$ given in \Cref{eq:walk-extensor-sum} is a nonzero polynomial in the $y$ variables if and only if $G$ contains a $k$-path.
	\end{lemma}
	
	We now recall the classical DeMillo-Lipton-Schwartz-Zippel Lemma:
	\begin{lemma}[\cite{schwartz-zippel-1,schwartz-zippel-2,schwartz-zippel-3}] \label{lem:schwartz-zippel}
	    Let $f\in \FF[x_1,...,x_n]$ be a nonzero polynomial in $n$ variables with total degree at most $d$ over some field $\FF$, and let $S\subseteq \FF$. For $a_1,a_2,...,a_n\in S$ chosen uniformly at random, we have $\Pr(f(a_1,a_2,...,a_n)=0)\le\frac{d}{|S|}$.
	\end{lemma}
	
	Using the DeMillo-Lipton-Schwartz-Zippel Lemma, we now obtain a randomized algorithm for the \textsc{$k$-Path} detection problem, assuming we are able to efficiently compute $\Z$. Namely, for each $e\in E(G)$ we randomly select $y_e\in Y$ for some fixed $Y\subseteq \FF$ of size $|Y|=100k$, and compute $\Z$ as in \Cref{eq:walk-extensor-sum}. If this value is nonzero, we know there is a $k$-path. Otherwise, there might still be a $k$-path, and we might have been unlucky and had the nonzero polynomial vanish for the specific choices of the $y_e$ variables. We output that there is no $k$-path in this case. The probability of error is at most $\frac{k}{|Y|}= \frac{1}{100}$, and we can repeat this process to lower the probability of error as much as required.
	
	It remains to show how to compute $\Z$ efficiently. We do this using dynamic programming. For any $1\le s \le k$ and $1\le i \le n$ we define $Q_s[i]$ as the sum of walk extensors of length $s$ that end in $v_i$. Then $\Z=\sum_i Q_k[i]$, and we can compute the vector $Q_{s+1}$ from $Q_s$ by $Q_{s+1}[i] = \sum_{j:(v_j,v_i)\in E(G)} Q_s[j] \wedge y_{v_j, v_i} \wedge \chi(v_i)$. This requires $kn^2$ skew multiplications of extensors, each of which can be done in time $2^k \poly(k)$ (see \Cref{par:skew_prod}). Thus, we can compute $\Z$ with $2^k \poly(k) n^2$ field operations, which is also the total running time of the randomized algorithm. We note that this works over any sufficiently large field $\FF$ with $|\FF| \geq 100k$.
	
	Brand, Dell and Husfeldt~\cite{extensor-coding} also give a deterministic variant of the algorithm, at the cost of increasing the time complexity. This is done with the beautiful idea of, for each vertex $v$, \emph{`lifting'} the Vandermonde vector $\chi(v)$ to $\bar{\chi}(v)=\binom{\chi(v)}{0}\wedge \binom{0}{\chi(v)}$. By $\binom{\chi(v)}{0}$, we mean the vector of length $2k$ gotten by concatenating the vector $\chi(v)$ with $k$ zeros. We then do the calculations in $\Lambda(\FF^{2k})$ instead of $\Lambda(\FF^{k})$. Walk extensors of non-path walks still vanish. Now, as observed in \cite{extensor-coding}, for any $k$-path we have
	
	\[ \bar{\chi}(w_1)\wedge \bar{\chi}(w_2) \wedge ... \wedge \bar{\chi}(w_k) = e_{[2k]} \det\left(\binom{\chi(w_1)}{0}\bigg|\binom{0}{\chi(w_1)}\bigg....\bigg|\binom{\chi(w_k)}{0}\bigg|\binom{0}{\chi(w_k)}\right) \]
	
	By proper change of columns, the resulting determinant is equal to the determinant of a $2\times 2$ block diagonal matrix, where the two diagonal blocks are identical. By basic properties of determinants, we then obtain
	
	\begin{equation}\label{eq:lifts_determinant}
	\bar{\chi}(w_1)\wedge \bar{\chi}(w_2) \wedge ... \wedge \bar{\chi}(w_k) = (-1)^{\binom{k}{2}}e_{[2k]} \det(\chi(w_1)| \chi(w_2) | ... |\chi(w_k))^2.
	\end{equation}
	
	Therefore, choosing $\FF=\QQ$, the walk extensors of different $k$-paths all have the same sign, and hence never cancel. The dynamic programming given above still works for this case, but extensor operations require $2^{2k}\poly(k)$ field operations (note we only need skew multiplications). Additionally, all numbers involved are of size at most $n^{O(k)} = 2^{O(k \log n)}$, and hence operations require only $\poly(k) \polylog(n)$ time, or $\poly(k)$ time in the word-RAM model which we assume. This produces a deterministic algorithm with running time $4^k\poly(k) n^2$.
	
	\section{Overview of our techniques} \label{sec:overview}
	The key idea behind our algorithms is to identify appropriate sums of extensors which encode the answer to the problem (e.g., one where candidate solutions correspond to nonzero monomials), and which can be efficiently dynamically maintained. Similar to~\cite{extensor-coding}, we will frequently make use of the properties of exterior algebras to `nullify' repetitions (for example, of vertices in the \textsc{$k$-Path} problem, and of elements in the \textsc{Exact $k$-Partial Cover} problem). In many cases, the extensors or similar algebraic coding expressions used by the fastest known static algorithms seem difficult to efficiently maintain dynamically, and we will need to modify them, and identify useful precomputed values, to speed up our update time.
	While these techniques prove to be very well suited for designing sensitivity oracles and dynamic algorithms, we are unaware of previous work that uses the exterior algebra in such a way.
	
	\subsection{\textsc{\texorpdfstring{$k$}{k}-Path}}
	For this problem, we aim to maintain $\Z$, the sum over walk extensors of length-$k$ walks which we defined in~\Cref{eq:walk-extensor-sum} above.
	Employing properties of the exterior algebra as discussed above, it is seen that any walk that is not a path (that is, a walk that repeats a vertex) does not contribute to the sum, and we thus indirectly compute a sum over only paths.
	
	As we discussed in \Cref{sec:extensor-coding} above, Brand, Dell and Husfeldt~\cite{extensor-coding} showed that computing this sum allows for extracting valuable information about the $k$-paths in a graph. In particular, by carefully selecting vectors or degree-2 extensors $\chi(v)$ for each vertex $v$, it allows for the design of randomized and deterministic algorithms for the $k$-path detection problem, as well as for approximately counting the number of $k$-paths with high probability.
	
	Let us focus, first, on the case when our sensitivity oracle only allows for edge increments. 
	We begin by precomputing, for each pair of vertices in the graph, a sum over the walk extensors between those vertices, so that we can `stitch' them together appropriately when new edges are inserted. Put precisely, let $W_s(u,v)$ denote the set of walks of length $s$ from $u$ to $v$, and let $Q_s[u,v]$ be the precomputed value of the sum of walk extensors for walks in $W_s(u,v)$ in the initial graph. Now suppose a new edge $(v_1,v_2)$ is inserted. Then the additional walk extensors for walks of length $s$ between any pair of vertices $(t_1,t_2)$ can be computed as
	\begin{align*}
	\sum_{s'}\left(\sum_{w\in W_{s'}(t_1,v_1)}\chi(w)\right)\wedge y_{v_1,v_2}\wedge \left(\sum_{w\in W_{s-s'}(v_2,t_2)}\chi(w)\right) \\
	= \sum_{s'} Q_{s'}[t_1,v_1]\wedge y_{v_1,v_2}\wedge Q_{s-s'}[v_2,t_2] .
	\end{align*}
	At first, stitching might seem problematic, since we might need to iterate over the possible lengths of each stitched part. This is not be a problem for a single edge insertion, but leads to an exponential dependency on the number of edges inserted as we partition the total length $k$ between them. While this issue can be efficiently resolved with proper dynamic programming, we can circumvent this efficiently and more cleanly by instead working with sums over the walks of all possible lengths, while finally inspecting only the coefficient of $e_{[k]}$ in the resulting extensor, so that terms corresponding to paths of length other than $k$ will either cancel out or have too low degree. Indeed, we note that only walks of length $k$ contribute to this coefficient.
	
	Another, more substantial problem occurs when one tries to `stitch' paths into a larger path that uses several new edges: there are $\ell!$ different orders to pass through the $\ell$ new edges. We combine properties of the exterior algebra with an algebraic trick (which we describe in more detail shortly) to allow computing the sum of the effects of all of these cases with only a polynomial dependence on $\ell$.
	
	In order to be able to work with only a small subset of the precomputed $Q$ matrix, we additionally precompute the sum of walk extensors ending and beginning at each vertex (i.e., sums of rows and columns of $Q$), and the original sum over walk extensors in the original graph. This allows us to use only $O(\ell^2)$ precomputed values in the updating phase.
	
	Finally, by further employing properties of the exterior algebra, we are able to combine these techniques with the inclusion-exclusion principle, to also allow edge removals, while keeping the same time and space complexities for preprocessing and updates.
	
	We now describe the ideas in some more detail. We refer the reader to \Cref{sec:k-path-so} for the full details. As a preprocessing, we compute an $n\times n$ matrix $Q$ with extensor values, defined by
	\[ Q[i,j] = \sum_{\substack{(w_1,w_2,...,w_s) \in W \\ w_1=v_i,w_s=v_j}} \chi(w_1) \wedge y_{w_1w_2} \wedge \chi(w_2) \wedge y_{w_2w_3} \wedge \chi(w_3) \wedge ... \wedge y_{w_{s-1}w_s}\wedge \chi(w_s) \]
where $W$ is the set of \emph{all walks} in the graph, which can have any length greater than 0. This is the sum of all walk extensors for walks from $v_i$ to $v_j$. We then compute vectors $S$ (resp. $F$) of length $n$, similarly computing in their $i$th entry the sum of walk extensors starting (resp. finishing) at each vertex $v_i$. That is, $S[i]=\sum_j Q[i,j]$ and $F[j] = \sum_i Q[i,j]$.
	
	Finally, we compute $\Z=\sum_{i,j} Q[i,j]$, the sum of all walk extensors over all walks. $e_{[k]}\Z$ is exactly the value we aim to maintain after a query, as this is exactly the one used in the extensor coding technique (\Cref{sec:extensor-coding}).
    We explain how with proper dynamic programming we can compute these values in preprocessing with $2^k \poly(k) n^2$ operations over $\FF$, using only additions and skew multiplications.

    Then, when prompted with a query, we aim to compute the sum over walk extensors in the updated graph, denoted $\Znew$. Inspecting whether $e_{[k]}\Znew \neq 0$ will let us test whether the updated graph contains a $k$-path.

    We now begin handling a query. Given a list of $\ell$ updates that are each either an edge insertion or edge deletion, there are $n' \leq 2 \ell$ vertices at the endpoints of the updated edges. We begin by extracting the $n' \times n'$ sub-matrix $Q'$ of $Q$, and length-$n'$ sub-vectors $S', F'$ of $S,F$, gotten by taking entries corresponding to those endpoint vertices.

    We next define an $n' \times n'$ matrix $E^+_r$ corresponding to the $r$-th edge insertion by setting each of its entries to $0$ except for $E^+_r[i,j]=y_{i,j}$, where $(v_i,v_j)$ is the $r$-th inserted edge. We similarly define $E^-_r$ in the same way, to be the 0 matrix except for $E^-_r[i,j]=y_{i,j}$ where $(v_i,v_j)$ is the $r$-th deleted edge. We then define $\Delta^+=\sum_r E^+_r$, $\Delta^-=\sum_r E^-_r$, and $\Delta=\Delta^+ - \Delta^-$. We then compute that
    \[ \Znew = \Z + \sum_{i=1}^k F'^T\Delta(Q'\Delta)^{i-1} S'.\]
    This formula is crucial to our query algorithm, so we explain it in some detail here.

    We first explain the correctness of this formula for the case of only increments. In this case, $\Znew - \Z$ is exactly the sum of walk-extensors for walks that use the new edges.
	Here, upon expanding the expression when $\sum_r E^+_r$ is substituted for $\Delta^+$, we note that $F'^T\Delta^+S'$ counts the walks that use exactly one of the new edges. Indeed, for each $r$ the term $F'^T E_r^+S'$ accounts for walks passing through the $r$-th new edge exactly once. For counting walks that use exactly two of the newly inserted edges, we now need to compute $\sum_{r_1\neq r_2}F'^T E_{r_1}^+Q'E_{r_2}^+S'$. Indeed, this expression correctly accounts for walks with nonzero walk extensors starting at any vertex, continuing through one new edge, travelling to a new edge then through it, then continuing to end at any vertex. We further note that any specific path is counted in this way exactly once.
	
	We note that by linearity and the fact that only paths produce nonzero walk extensors, this is equal to $F'^T\Delta^+Q'\Delta ^+ S'$. That is, upon expanding the expression when $\sum_r E^+_r$ is substituted for $\Delta^+$, we get exactly the terms we intended, plus terms accounting for walks passing through the new edges twice, which evaluate to zero in the exterior algebra. This considerably simplifies the calculations.
	Continuing in the same fashion, we argue we obtain
	\[ \Znew - \Z = \sum_{i=1}^k F'^T\Delta^+(Q'\Delta ^+)^{i-1} S'\]
	in total for the incremental case.
	
	Consider now the general case, with both increments and decrements. Suppose $w$ is any walk on the vertices which uses edges from the union of the original graph and the updated graph. We assume $w$ is a path, since otherwise its walk extensor vanishes, and we need not worry about whether it appears in the sum $\Znew$. Now suppose $w$ has length at most $k$ and uses $a$ inserted edges and $b$ removed edges. 
	
	 Consider first when $a\ge 1$. In this case, $w$ is not counted in the original sum $\Z$. When substituting $\Delta=\sum_{j} E_j^+ - \sum_j E_j^-$ into $\sum_{i=1}^k F'^T\Delta(Q'\Delta)^{i-1} S'$ and expanding, we see that $w$ is counted only when all the chosen $E_j^+$ factors exactly correspond to the $a$ new edges used by $w$, in the order they are used. It is also counted when choosing any $b'\le b$ factors of type $E_j^-$ that appear in $w$, but they must also appear in the right order, and they contribute the walk extensor of $w$ with weight exactly $(-1)^{b'}$. In total, $w$ is accounted for exactly $\sum_{b'=0}^b \binom{b}{b'}(-1)^{b'} = (1-1)^b = [b=0]$ times\footnote{Here we use the notation $[b=0] := \begin{cases}
  1  & \text{if } b=0, \\
  0 &\text{otherwise.}
\end{cases}$} in $\Znew$, which is what we want.
	
	Now suppose $a=0$. Then a similar argument shows that the number of times it is counted in $\sum_{i=1}^k F'^T\Delta(Q'\Delta)^{i-1} S'$ is exactly $\sum_{b'=1}^b \binom{b}{b'}(-1)^{b'} = (1-1)^b - 1 = [b=0] - 1$. However, it is also counted in $\Z$ exactly once, so in total is counted $[b=0]$ times, which is also exactly what we want.
    
    Using this formula in \Cref{sec:k-path-so} below, we explain how we are generally able to compute $\Znew$ in $O(\ell^2 2^{\omega k/2})$ field operations. However, we come back to this running time once we describe the specifics of the randomized and deterministic sensitivity oracles, each with its own details.
	
	\subsection{\textsc{\texorpdfstring{$k$}{k}-Partial Cover}}
	We also make use of extensors to design fully dynamic algorithms for other parameterized problems. We focus here on one example: designing a deterministic dynamic algorithm for the \textsc{$k$-Partial Cover} problem. In this problem, we are given subsets $S_1,...,S_n\subseteq [N]$, and wish to find the minimum number of such subsets whose union has size at least $k$ (or report that no such collection exists).
	
	One could use a polynomial constructed by Koutis and Williams~\cite{algebraic-problems} for this problem, combined with similar techniques to those we used above for \textsc{$k$-Path}, to devise a deterministic dynamic algorithm with update time $O(2^{\omega k})$. However, we instead give a different polynomial which is easier to dynamically maintain, achieving a faster $O^*(4^k)$ update time.
	
	As in the \textsc{$k$-Path} problem, we give each element in the universe $a\in [N]$ a Vandermonde vector $\chi(a)\in\QQ^k$, then lift it to $\bar{\chi}(a)=\binom{\chi(a)}{0}\wedge\binom{0}{\chi(a)}\in\Lambda(\QQ^{2k})$.
	
	We introduce a single new variable $z$, and consider the polynomial
	\[ P(z) = \prod_S \left(1 + z\left[\prod_{a\in S} (1+\bar{\chi}(a)) - 1\right] \right).\]
	We argue that the solution we seek is the minimum $t$ for which $[e_{[2k]}z^t]P\neq 0$ (that is, the coefficient of $e_{[2k]}$ in the coefficient of $z^t$ in $P$). Our goal, then, will be to maintain the value of $P(z)$, and calculations will take place in polynomials over extensors, $\Lambda(\QQ^{2k})[z]$. We also argue that $\deg P(z) \le k$, and so is not too large to handle in update steps.
	
	We first argue that the product in $P$ should only be computed over sets $S$ of cardinality less than $k$, as larger sets can be treated separately, knowing that the optimal solution is 1 if any such set exists in the collection. Since $P(z)$ is defined as a product over sets, to update it with a new set $S$ of cardinality $|S|<k$, we can multiply $P(z)$ by the corresponding factor $1 + z\left[\prod_{a\in S} (1+\bar{\chi}(a)) - 1\right]$. This seems too slow at first, since general extensor multiplication requires $O(2^{\omega k})$ time, which is more than our target $O^*(4^k)$. However, rather than computing the factor and then using general extensor multiplication, we argue that we can indirectly multiply by this factor using only additions and skew multiplications by rearranging the contributing terms, thus requiring only $O^*(4^k)$ field operations.
	Indeed, the product of $P(z)$ with the new factor is $(1-z)P(z) + P(z)\prod_{a\in S} (1+\bar{\chi}(a))$, which can be performed by separately computing $(1-z)P(z)$ and $P(z)\prod_{a\in S} (1+\bar{\chi}(a))$, where the latter can be computed by repeated skew-multiplications.
	
	A new problem arises when a set $S$ is removed. For this, we need to somehow cancel the factor in $P(z)$ corresponding to $S$. We show that in this case the corresponding factor has an inverse in $\Lambda(\QQ^{2k})[z]$: we first note that it can be written as $1+X$ for an element $X$ that contains only extensors of degree $\ge 2$, which implies that $X^{k+1}=0$, and using this, we observe that
	\[ (1+X)(1-X+X^2-\ldots+(-X)^k)=1-(-X)^{k+1}=1,\]
	so $(1+X)^{-1} = 1-X+X^2-\ldots+(-X)^k$. Similar to the previous case, we argue that multiplying by this factor can also be done with $\poly(k)$ additions and skew multiplications, and so can be done in $O^*(4^k)$ time.
	
	This and all our other dynamic algorithms are described in detail in \Cref{sec:dynamicalgs}.

\section{Sensitivity Oracle for the \textsc{\texorpdfstring{$k$}{k}-Path} problem (allowing both increments and decrements)} \label{sec:k-path-so}

We give two versions of the algorithm - one randomized and one deterministic, with different time bounds.

\begin{theorem} \label{thm:kpath-so-rand}
	For the \textsc{$k$-Path} detection problem on a directed graph $G$ on $n$ vertices and with $m$ edges, there is a randomized sensitivity oracle with preprocessing time $2^k \poly(k) \min\{nm, n^\omega\}$ and query time $\ell^2 2^k \poly(k)$.
\end{theorem}

\begin{theorem} \label{thm:kpath-so-det}
	For the \textsc{$k$-Path} detection on a directed graph $G$ on $n$ vertices and with $m$ edges, there is a deterministic sensitivity oracle with $4^k\poly(k) \min\{nm, n^\omega\}$ preprocessing time and $O(\ell^2 2^{\omega k})$ query time.
\end{theorem}

\begin{corollary} \label{cor:kpath-so-undirected}
	For the \textsc{$k$-Path} detection on an \emph{undirected} graph, there is a randomized sensitivity oracle with $2^k \poly(k) n^2$ preprocessing time and $\ell^2 2^k \poly(k)$ query time, and a deterministic sensitivity oracle with $4^k \poly(k) n^2$ preprocessing time and $O(\ell^2 2^{\omega k})$ query time.
\end{corollary}

We mention that it is also possible to allow vertex decrements with the same asymptotic times (i.e., with only a constant factor overhead). This is contrasted with the solution in \cite{so-decremental} which reduces to the $2k$-path problem, hence doubling the value of $k$ in the time bounds formulas. This is described in \Cref{sec:vertex-failures}.

We begin by providing the common framework for the randomized and deterministic versions, then describe the necessary adjustments for each of them.

\Cref{cor:kpath-so-undirected} then follows by noting that in an undirected graph with $m>nk$, there is always a $k$-path \cite{edge-bound}, hence if $m>nk+\ell$ there should be no preprocessing, and at update time the result is always `Yes'. We implicitly assume $\ell<n$, as otherwise it would be better to just compute the result from scratch.

\subsection{Preprocessing}
Given a directed graph $G$, we denote the set of vertices by $V=\{v_1,...,v_n\}$, and we denote by $W$ the set of \emph{all walks} in the graph, which can have any length greater than 0. Recall that walks may repeat vertices, whereas paths may not. For each pair $\{i,j\} \in [n]^2$ we have an \emph{edge variable} $y_{v_i,v_j}$, whose possible values will be in $\FF$. Note we also define variables for edges not presented in the initial input graph, because new edges are allowed to be inserted.

Furthermore, we define vectors $\chi:V\to \FF^k$ where $\FF$ is a field to be chosen later, by $\chi(v_i)=(1, j, j^2, \ldots, j^{k-1})$ where $j = f(i)$ for some injective $f : [n] \to \FF$. We call these \emph{Vandermonde vectors}.

In the preprocessing phase, we compute the following values.

\begin{itemize}
	\item An $n\times n$ matrix $Q$ with extensor values, such that
	\begin{equation}
	    Q[i,j] = \sum_{\substack{(w_1,w_2,...,w_s) \in W \\ w_1=v_i,w_s=v_j}} \chi(w_1) \wedge y_{w_1w_2} \wedge \chi(w_2) \wedge y_{w_2w_3} \wedge \chi(w_3) \wedge ... \wedge y_{w_{s-1}w_s} \wedge \chi(w_s)
	\end{equation}
	that is, the sum of extensors along all walks that start with $v_i$ and finish with $v_j$.
	It is crucial to note that the sum is in fact finite, because any walk of length more than $k$ will contribute 0 to the sum.
	\item Vectors $S$ and $F$ of length $n$ with extensors values, corresponding to all walks that start (resp. finish) with each specific vertex. So
	\begin{align}
	    S[i] = \sum_{j=1}^n Q[i,j] \\
	    F[j] = \sum_{i=1}^n Q[i,j]
	\end{align}
	\item An extensor $\Z$ that is the sum of all walk extensors over all walks. Thus, $\Z = \sum_{i,j} Q[i,j]$.
\end{itemize}

In order to compute $Q$, we iteratively compute $Q_s$ where $Q_s$ is defined in the same way as $Q$, but for walks of length exactly $s$. We will finally compute $Q=\sum_{s=1}^k Q_s$. This is done by noting that $Q_1$ is diagonal with $Q_1[i,i]=\chi(v_i)$, and that a walk from $v_i$ to $v_j$ of length $s+1$ can be described as a walk of length $s$ from $v_i$ to some other $v_r$ followed by an edge from $v_r$ to $v_j$, giving $Q_{s+1} = Q_s YQ_1$, where $Y[i,j]=y_{i,j}$ if $(v_i,v_j)\in E(G)$ and 0 otherwise. Thus, denoting $D=Q_1$, we have
\begin{equation}
    \forall s\ge 1: Q_s = D(YD)^{s-1} .
\end{equation}

Since a general product of two extensors in $\Lambda(\FF^k)$ can be done in $O(2^{\omega k / 2})$ field operations (\Cref{par:general-prod}), and since matrix multiplication can be done with $O(n^\omega)$ ring operations over any (possibly noncommutative) ring, we can compute $Q$ with $O(n^\omega 2^{\omega k/2})$ operations over $\FF$.

However, we can do this even faster, by iteratively multiplying by $D$ and $Y$ separately. Multiplication by $D$ requires only $n^2$ multiplications of extensors by mere vectors (that is, skew multiplications), and hence can be done in time $2^k \poly(k) n^2$ (\Cref{par:skew_prod}). Multiplication by $Y$ requires multiplication of extensors only by scalars, and hence can be done in time $O(2^k)$ per operation as well. Here, we can either use general matrix multiplication algorithm in $n^{\omega+o(1)}$ field operations, or sparse matrix multiplication in $O(nm)$ operations.

\subsection{Computing result after a query} \label{sec:k-path-query}
By following the techniques explained in \Cref{sec:extensor-coding}, we note that it is enough to compute the sum of walk extensors in the updated graph. However, compared to the static algorithm, we maintain the sum over all walks, and not just walks of length $k$. This is remedied, when answering queries, by extracting the coefficient of $e_{[k]}$ in $\Z$, to which only the length-$k$ walks contribute.

As a warm up, consider adding a single edge first, $(v_1, v_2)$. Then all the walks that contributed to $\Z$ still exist, but we added some new walks. Specifically, these are the walks that include the new edge at least once. However, walk extensors for walks that include the edge more than once are not contributing to the sum, and hence the only new walks that have to be considered are those that can be decomposed by a walk ending in $v_1$, then traversing the edge $(v_1,v_2)$, then a walk starting at $v_2$. Thus, the new walk sum is $\Z_{\text{new}}=\Z+F[1]\wedge y_{1,2} \wedge S[2]$.

For the general case, we are given a list of $\ell$ updates that are each either an edge insertion or edge deletion. We assume no edge appears twice in this list, as we can just cancel an edge insertion + deletion. Let $I$ be the list of indices $i$ for which $v_i$ is an endpoint of one of the updated edges. We see that $|I|\le 2\ell$. Further, let $Q'=Q[I,I]$ be the sub-matrix $Q$ where we keep only walks between two vertices in $I$. Similarly define $F'=F[I]$ and $S'=S[I]$. Without loss of generality, we assume $I=\{1,...,n'\}$ where $n'=|I|\le 2\ell$.

We next define an $n' \times n'$ matrix $E^+_r$ corresponding to the $r$-th edge insertion by setting each of its entries to $0$ except for $E^+_r[i,j]=y_{i,j}$, where $(v_i,v_j)$ is the $r$-th inserted edge. We similarly define $E^-_r$ to be the 0 matrix except for $E^-_r[i,j]=y_{i,j}$ where $(v_i,v_j)$ is the $r$-th deleted edge. We then define $\Delta^+=\sum_r E^+_r$, $\Delta^-=\sum_r E^-_r$, and $\Delta=\Delta^+ - \Delta^-$. We claim that

\begin{lemma}\label{lem:Z-new}
	The value of the updated $\Z$, following the edge insertions and deletions, is
	\begin{equation}
	    \Znew = \Z + \sum_{i=1}^k F'^T\Delta(Q'\Delta)^{i-1} S'
	\end{equation}
\end{lemma}

The proof is given below. Using this, we can show:

\begin{lemma}\label{lem:Z-new-complexity}
	$\Znew$ for a given set of updates can be computed from $Q$, $F$, $S$ and $\Z$ in $O(\ell^2 2^{\omega k/2})$ field operations.
\end{lemma}
\begin{proof}[Proof of \Cref{lem:Z-new-complexity}]
	Using \Cref{lem:Z-new}, it follows that we can compute $\Znew$ as follows. We first compute $Q' \Delta$ with simple sparse matrix multiplication, noting that $\Delta$ has at most $2\ell$ nonzero entries. Only $O(\ell^2)$ multiplications of general extensor expressions are needed for this, which therefore can be done with $O(\ell^2 2^{\omega k/2})$ field operations (employing \Cref{par:general-prod}). We then iteratively multiply this product by the vector $S'$, which again requires $O(\ell^2 2^{\omega k/2})$ operations per multiplication (employing \Cref{par:general-prod} again). Adding all the resulting vectors, then multiplying on the left by $\Delta$ then by $F'^T$ requires the same amount of time, totalling $O(\ell^2 2^{\omega k/2})$ operations for the whole update.
\end{proof}

\begin{proof}[Proof of \Cref{lem:Z-new}]
	
	As a warm-up, we consider first two easier cases: the only-increments case, and the only-decrements case.
	
	In the only-increments case, $\Znew - \Z$ is exactly the sum of walk-extensors for walks that use the new edges. In the only-decrements case, this is the walks that used the removed edges (now with the opposite sign).
	
	\begin{itemize}
		\item Only increments: similar to the warm up example, we note that $F'^T\Delta^+S'$ counts the walks that use exactly one of the new edges. For counting walks that use exactly two of the newly inserted edges, we now need to compute $\sum_{r_1\neq r_2}F'^T E_{r_1}^+Q'E_{r_2}^+S'$. Indeed, this expression correctly accounts for walks with nonzero walk extensors starting at any vertex, continuing through one new edge, travelling to a new edge then through it, then continuing to end at any vertex. We further note that any specific path is counted in this way exactly once.
	
    	We note that by linearity and the fact that only paths produce nonzero walk extensors, this is equal to $F'^T\Delta^+Q'\Delta ^+ S'$. That is, upon expanding the expression when $\sum_r E^+_r$ is substituted for $\Delta^+$, we get exactly the terms we intended, plus terms accounting for walks passing through the new edges twice, which evaluate to zero in the exterior algebra.
		
		Continuing in the same fashion, we get
		\begin{equation}
		    \Znew - \Z = \sum_{i=1}^k F'^T\Delta^+(Q'\Delta ^+)^{i-1} S'
		\end{equation}
		where the sum is up to $k$ because, if there are less than $k$ new edges then the sum is unchanged by adding these non-path walks, and if there are more than $k$ new edges then clipping the upper limit at $k$ only removes walks of length more than $k$, which evaluate to zero because we're computing in $\Lambda\left(\QQ^k\right)$.
		
		\item Only decrements: we use inclusion-exclusion. Let $W_i$ be the set of original walks using the $i$th removed edge, and for a set $U$ of walks we let $|U|$ denote the sum of walk extensors over the walks in $U$. Then
		\begin{equation}\Z^- = \Big|\bigcup_i W_i\Big| = \sum_i |W_i| - \sum_{i<j} |W_i \cap W_j| + \sum_{i<j<r} |W_i \cap W_j \cap W_r| - ...\end{equation}
		Similar to the computation of $\Z^+$, we see that $|W_i| = F'^TE_i^-S'$, and hence $\sum_i |W_i| = F'^T \Delta^- S'$.
		
		For $i\neq j$ computing $|W_i \cap W_j|$ is more subtle. We first claim that
		\begin{equation}|W_i \cap W_j| = F'^T E_i^- Q' E_j^- S' + F'^T E_j^- Q' E_i^- S' \end{equation}
		
		which comes from the fact that the first term on the right hand side counts the walks that pass through the $i$-th edge before the $j$-th edge, and the second term counts the other case. To see this, note that since we're computing with extensors, each walk that evaluates to a nonzero extensor will be counted exactly once.
		
		As before, we can now see that the sum over these is simply
		\begin{equation}\sum_{i<j}|W_i \cap W_j| = F'^T \Delta^- Q' \Delta^- S' \end{equation}
		
		because by substituting $\Delta^-=\sum_i E_i^-$ into the right hand side and expanding, we recover each of the terms exactly once, and the additional terms of repeated edges evaluate to zero.
		
		In a similar fashion, we get $\sum_{i<j<r} |W_i \cap W_j \cap W_r| = F'^T \Delta^- Q' \Delta^-Q' \Delta^- S' $ and so on, giving altogether
		\begin{equation}\Znew - \Z = \sum_{i=1}^k  F'^T (-\Delta^-) (-Q' \Delta^-)^{i-1} S'. \end{equation}
	\end{itemize}
	
	Putting it all together, we will show that $\Z + \sum_{i=1}^k F'^T\Delta(Q'\Delta)^{i-1} S'$ counts all the walk extensors in the updated graph exactly once, and does not count any other walks.
	
	Indeed, suppose $w$ is any walk on the vertices (not constrained to be in the original or updated graphs). If $w$ uses any edge that is neither in the original graph nor the updated graph, then it is not counted, so we disregard these cases. We assume $w$ is a path, since otherwise its walk extensor vanishes anyway, and we can regard it as accounted for. Now suppose $w$ has length at most $k$ and uses $a$ new edges and $b$ removed edges. 
	
	We first assume $a\ge 1$. In this case, $w$ is not counted in $\Z$. When substituting $\Delta=\sum_{j} E_j^+ - \sum_j E_j^-$ into $\sum_{i=1}^k F'^T\Delta(Q'\Delta)^{i-1} S'$ and expanding, we see that $w$ is counted only when all the chosen $E_j^+$ factors exactly correspond to the $a$ new edges used by $w$, in their order. It is also counted when choosing any $b'\le b$ factors of type $E_j^-$ that appear in $w$, but they must also appear in the right order, and they contribute the walk extensor of $w$ with weight exactly $(-1)^{b'}$. In total, $w$ is accounted for exactly $\sum_{b'=0}^b \binom{b}{b'}(-1)^{b'} = (1-1)^b = [b=0]$ times, which is what we want.
	
	Now suppose $a=0$. Then a similar argument shows that the number of times it is counted in $\sum_{i=1}^k F'^T\Delta(Q'\Delta)^{i-1} S'$ is exactly $\sum_{b'=1}^b \binom{b}{b'}(-1)^{b'} = (1-1)^b - 1 = [b=0] - 1$. However, it is counted in $\Z$ exactly once, so in total is counted $[b=0]$ times, which is also exactly what we want.
\end{proof}

Although not useful for our purposes, we note that we are also able to compute $Q'_\text{new}$ after the updates, i.e., the matrix where $Q'_\text{new}[i,j]$ is the sum of walk extensors of walks from $i$ to $j$ in the whole updated graph, where $i,j$ are restricted to being in $I$. This can conceivably be useful for other problems, and for example can be used to answer queries of the type ``is the distance from $i$ to $j$ at most $k$ in the updated graph''. The calculation is similar to the one in the proof of \Cref{lem:Z-new-complexity}, and is given by
\begin{equation}Q'_\text{new} = \sum_{i=0}^k Q' (\Delta Q')^i. \end{equation}

\subsection{Adjustments for the randomized case}
In the randomized setting, we note that it is enough to choose $\FF$ to be a large-enough field of characteristic 2. Specifically, $|\FF|\ge 100k$ is enough, then the $y$ variables have the required domain size to use the DeMillo-Lipton-Schwartz-Zippel Lemma (\Cref{lem:schwartz-zippel}).

Here $\chi:V\to \FF^k$ can be chosen uniformly at random instead of Vandermonde vectors, and for all edges $(u,v)$ we have $y_{u,v}\in \FF$ uniformly at random.

The crucial point is that over characteristic 2, the exterior algebra becomes commutative. This allows computing the product of two extensors in $\Lambda(\FF^k)$ in time $2^k \poly(k)$, using fast subset convolution (see \Cref{par:general-prod-char-2}). As originally observed in \cite{extensor-coding}, this algebra is then isomorphic to the one used by Williams~\cite{williams-k-path} to solve the \textsc{$k$-Path} problem in $O^*(2^k)$ time for the first time (though here this is over a slightly different field).
This proves \Cref{thm:kpath-so-rand}.

\subsection{Adjustments for the deterministic case}
In this case we choose $\FF=\QQ$.
In the same way as \cite{extensor-coding} describe how to transform their algorithm into a deterministic one, we can transform the algorithm presented here as well. Specifically, all $y_{i,j}$ are set to 1, and all $\chi(v_i)$ are replaced with $\bar{\chi}(v_i) = \binom{\chi(v_i)}{0}\wedge \binom{0}{\chi(v_i)}\in\Lambda(\QQ^{2k})$. As explained in \Cref{sec:extensor-coding}, this ensures that extensor walks of different paths never cancel, and hence the coefficient $[e_{[2k]}]\Znew$ will be nonzero exactly if there is a $k$-path.

Since the exterior algebra is $\Lambda(\QQ^{2k})$, general multiplication requires $O(2^{\omega k})$ operations over $\QQ$ (see \Cref{par:general-prod}).
It can be seen that the numbers we're dealing with are integers of size at most $n^{O(k)}=2^{O(k \log n)}$, and thus each $\QQ$ operation requires $\poly(k) \polylog(n)$ time, or $\poly(k)$ time in the word-RAM model.
The rest of the details is essentially the same. This proves \Cref{thm:kpath-so-det}.

\subsection{Supporting vertex failures}
\label{sec:vertex-failures}
As mentioned, we can allow vertex decrements with no asymptotic overhead (i.e., the running times are at worst multiplied by a constant), in both the randomized and deterministic cases. This is done by first transforming the original graph, splitting each vertex into $v_\text{in}$ and $v_\text{out}$, as suggested in \cite{so-decremental}, where all inbound edges of $v$ are redirected to $v_\text{in}$ and outbound edges redirected to originate at $v_\text{out}$, and an additional edge $(v_\text{in}, v_\text{out})$ is inserted.

We then put $\chi(v_\text{in})=1$ for all $v$, while $\chi(v_\text{out})$ remains the original $\chi(v)\in\FF^k$ (in the randomized case) or $\bar{\chi}(v)\in\QQ^{2k}$ (in the deterministic case). Now, there is a $k$-path in the original graph if and only if there is a $2k$-path in the transformed graph, if and only if there is a path through $k$ out-vertices. So all calculations remain in the same exterior algebra of the same dimension, and the rest is largely unchanged. Removing a vertex $v$ is then done by removing the edge $(v_\text{in}, v_\text{out})$ in the transformed graph.

\subsection{Space lower bounds}
In \cite{so-decremental} the authors present a fault-tolerant algorithm for the \textsc{$k$-Path} problem with oracle size (space usage) $o(n)$ when keeping the parameters $k$ and $\ell$ constant (in fact, at most $O(\log n)$ in the two alternatives they present).

In contrast, the solution presented here requires $n^2 2^k \poly(k)$ bits. We show that for our case of allowing both increments and decrements, we must use $\Omega(n^2)$ space. Although storing the original graph clearly requires $\Omega(n^2)$ bits, we note that it is not necessary to store the original graph, so long as we assume that all provided updates are valid (i.e., we do not need to check that only non-existing edges were added and that removed edges were previously present).

Formally:

\begin{theorem}\label{thm:so-space-directed} Any deterministic sensitivity oracle for the \textsc{$k$-Path} problem on a directed graph, supporting at least 2 increments and $k\ge 4$, must store at least $\Omega(n^2)$ bits.
This bound also applies to any randomized sensitivity oracle with success probability bounded away from 0.
\end{theorem}

This bound is optimal from an information-theoretic point of view, because we can store the adjacency matrix with $O(n^2)$ bits.

\begin{proof}[Proof of \Cref{thm:so-space-directed}]
	We note it is enough to prove the claim for the case of $\ell=2$ and $k=4$.
	
	Consider the set of all directed bipartite graphs, with sides $S, T$, such that $|S|=|T|=\frac{n-2}{2}$ and all edges are from a vertex in $S$ to a vertex in $T$. There are $2^\frac{(n-2)^2}{4}$ such graphs. Add two special vertices $s$ and $t$, not connected to any vertex.
	
	Now suppose there is a deterministic sensitivity oracle storing less than $\frac{(n-2)^2}{4}$ bits. Then there will be two such bipartite graphs $G_1$, $G_2$ with the same oracle. Without loss of generality, suppose $G_1$ has an edge $(u,v)$ that is not present in $G_2$. Consider the update of adding the edges $(s,u)$ and $(v, t)$. It is seen that in $G_1$ there is a path of length 4 ($s\to u\to v\to t$), but there is no such path in $G_2$. However, the oracle must give the same output on both of these graphs, a contradiction. This shows the lower bound for the deterministic case.
	
	We move to the case of a randomized sensitivity oracle. While a randomized sensitivity oracle does not need to be able to distinguish between all the above $2^{(n-2)^2/4}$ graphs, we will show that the existence of one implies the existence of a deterministic sensitivity oracle distinguishing a subset of size $2^{\Omega(n^2)}$ of the above graphs.
	
	By a simple amplification argument, we can assume each query has success probability of at least $99\%$ (since by assumption it is bounded away from 0). We use a simple claim to finish the proof.
	
	\begin{claim*}
    For large $n$, there is a set of $2^{n^2/100}$ different graphs of the above type with more than $\frac{1}{10}n^2$ differing edges.
    \end{claim*}
    \emph{Proof of claim}.
    This is shown by a simple probabilistic argument, noting that after enumerating the possible edges, we can identify any such graph with a bit vector of length $(n-2)^2/4$ whose $i$th entry marks whether the $i$th edge is present. We now randomly pick $2^{n^2/100}$ such bit vectors uniformly at random. The probability that two specific vectors among these have hamming distance at most $\frac{1}{20}n^2$ is at most
	\begin{align*} \Pr(\text{dist}\le n^2/20) &\le 2^{-(n-2)^2/4} \sum_{i=0}^{n^2/20} \binom{n^2/4}{i} \le 2^{-(n-2)^2/4} \frac{n^2}{20} \binom{n^2/4}{n^2/20} \\ & \le 2^{- 0.24 n^2} \left(e\frac{n^2/4}{n^2/20}\right)^{n^2/20} \le  2^{- 0.24 n^2} \left(5e\right)^{n^2/20} < 2^{-0.04 n^2},
	\end{align*}
	where we used the inequality $\binom{n}{k}\le \frac{n^k}{k!}\le \left(\frac{en}{k}\right)^k$. 
	Hence, with positive probability, all pairs among the $2^{n^2/100}$ drawn vectors are farther than $\frac{n^2}{20}$ from each other. In particular, such a set exists.
    $\blacksquare$
	
	We now fix such a set of graphs. Suppose we query the randomized sensitivity oracle on these graphs with queries used in the deterministic proof (that is, adding only $(s,u)$ and $(v,t)$ edges). By assumption, on expectation it will answer correctly at least 99\% of these queries on these graphs. Therefore, it is possible to fix the randomness used by the random oracle in a way that produces a deterministic oracle answering correctly at least 99\% of the above queries, and it follows that it can err on more than 10\% of the queries only on at most 10\% of the graphs. Thus, there are at least $\frac{9}{10}2^{n^2/100}=2^{\Omega(n^2)}$ graphs which are at distance at least $n^2/20$ from each other, for which the oracles answers answers at least 90\% of the queries correctly, that is, it errs on less than $\frac{1}{10}\cdot \frac{n^2}{4} = \frac{n^2}{40}$ queries on each graph. Thus, it can completely distinguish between these graphs, and hence it requires at least $\Omega(n^2)$ bits of memory.
\end{proof}

For undirected graphs, we have different a lower bound:

\begin{theorem}\label{thm:so-space-undirected} Any deterministic sensitivity oracle for the \textsc{$k$-Path} problem on an undirected graph, supporting at least 2 increments and $k\ge 4$, must store at least $\Omega(n \log n)$ bits of memory.
	This bound also applies to any randomized sensitivity oracle with success probability bounded away from 0.
\end{theorem}

We note that for this case, there is an information-theoretic (i.e., with unlimited running time on updates) upper bound of $O(k n \log n)$ bits, assuming $\ell = O(n)$. This follows from the fact that any undirected graph with $m\ge kn$ edges has a $k$-path, as shown in \cite{edge-bound}, and hence for $m < nk+\ell$ there will be a $k$-path after the update, regardless of the update. In this case, then, only $O(1)$ bits are required, and otherwise we store the adjacency matrix sparsely, using $O(\log n)$ bits per edge.

\begin{proof}[Proof sketch of \Cref{thm:so-space-undirected}]
	We use the same family of graphs described in the proof of \Cref{thm:so-space-directed}, but with undirected edges that make a perfect matching between the two sides $S$ and $T$. The rest of the proof is analogous, noting that the number of such graphs is $\left(\frac{n-2}{2}\right)! = 2^{\Omega(n \log n)}$.
\end{proof}

\subsection{Approximately counting \texorpdfstring{$k$}{k}-paths}
We note that our results can be naturally extended to approximately counting the number of $k$-paths, by computing the same extensor value used in Theorem 8 in Brand, Dell and Husfeldt \cite{extensor-coding} with a sensitivity oracle in the same fashion.

\begin{theorem}\label{thm:approx-count}
	For an initial graph $G$ on $n$ vertices and with $m$ edges, there exists a random sensitivity oracle that, given $\epsilon>0$ and $k$, produces an estimate to the number of $k$-paths that, with probability $>99\%$, is within $\epsilon$ relative error, with preprocessing time $\frac{1}{\epsilon^2}4^k\poly(k)\min\{nm, n^\omega\}$ and update time $O\left(\frac{1}{\epsilon^2} \ell^2 2^{\omega k}\right)$.
\end{theorem}

\begin{proof}[Proof sketch of \Cref{thm:approx-count}]
    We modify our algorithm in the same way as Brand, Dell and Husfeldt~\cite{extensor-coding} modified their algorithm to design an approximate counting algorithm, so we only sketch the details here.
	We choose a lifted $\bar{\chi}$ for $\chi:V\to \QQ^k$ chosen uniformly at random from $\chi(v)\in \{1,-1\}^k$, with all edges variables $y_{u,v}=1$. Our $[e_{[2k]}]\Znew$ can be seen to be equal to the sum over all walk extensors of length $k$ in the updated graph, which is the same value computed in \cite{extensor-coding}. Hence, as argued there, the expected value (over the randomness of $\chi$) of $[e_{[2k]}]\Znew$ is proportional to the number of $k$-paths, and the rest of their analysis carries over to our case, requiring $\frac{1}{\epsilon^2}$ applications of the same calculation to achieve the desired error bound.
\end{proof}

\subsection{Undirected graphs}
Björklund, Husfeldt, Kaski and Koivisto~\cite{narrow-sieves} give a randomized $1.66^k \poly(n)$-time algorithm for $k$-path detection in undirected graphs. We show how this algorithm can be transformed to an efficient sensitivity oracle.

\undirectedfastso*

The precise running time achieved by the algorithm is obtained by replacing $1.66^k$ in the above by $\alpha^{k+o(k)}$ where
\begin{equation}
    \alpha=2 \sqrt{2}^{\sqrt 2} (2-\sqrt{2})^{2-\sqrt 2} (\sqrt 2 - 1)^{\sqrt 2 - 1} \approx 1.656854
\end{equation}

First, we give a brief overview of the static algorithm. Following the observation of Brand, Dell, and Husfeldt~\cite{extensor-coding} that the algorithm can be rephrased in the language of extensors, we present the algorithm in this language. We note that in doing so, we lose the polynomial space complexity of the algorithm of \cite{extensor-coding}, making it exponential.

\subsubsection{Static randomized \texorpdfstring{$O^*(1.66^k)$}{O*(1.66\^k)} algorithm for \texorpdfstring{$k$}{k}-path detection on undirected graphs} \label{sec:narrow-sieves}

In the extensor-coding technique described so far, we assign a vector $\chi(v)\in \FF^k$ to each $v\in V$. Here, we instead assign vectors to a subset of the vertices and to a subset of the edges. We begin with randomly partitioning $V$ to two disjoint subsets, $V=V_1\cup V_2$, where each vertex is assigned to each part independently and with equal probability. Assuming there exists a $k$-path in the graph, we hope that one exists with exactly $k_1$ of its \emph{vertices} belonging to $V_1$, and $k_2$ of its \emph{edges} are connecting two $V_2$ vertices, where $k_1$ and $k_2$ are two new parameters with values $k_1=0.5k$ and $k_2=\big\lfloor \frac{\sqrt{2}-1}{2} k \big\rfloor \approx 0.207 k$. We will iterate the whole algorithm presented below with different random partitioning of $V$ into $V_1\cup V_2$ until the probability of successfully catching a $k$-path, if one exists, is high. Specifically, we repeat this for $\left(2(2-\sqrt{2})^{2-\sqrt 2} (\sqrt 2 - 1)^{\sqrt 2 - 1}\right)^{k+o(k)} < 1.015^k$ partitions.

We choose $\FF$ with characteristic 2 for our purposes, and work with $\Lambda(\FF^{k_1+k_2})=\Lambda(\FF^{\lfloor k/\sqrt{2}\rfloor })$, which is commutative. We assign a Vandermonde vector $\chi_1(v)\in \FF^{k_1}$ to each vertex in $V_1$ and extend it with $k_2$ zeros to obtain $\chi(v)=\binom{\chi_1(v)}{0}$, and similarly give each edge in $V_2\times V_2$ a Vandermonde vector of length $k_2$ and prepend it with $k_1$ zeros. We require that $\chi(uw)=\chi(wu)$ for any $u,w\in V_2$. Now, given a walk $W=v_1v_2...v_s$ on the graph, we define its walk extensor $\chi(W)$ differently from before, by:

\begin{equation}
    \chi(W) = z^s \prod_{i} y_{v_i,v_{i+1}}\cdot \prod_{i:v_i\in V_1} \chi(v_i) \cdot \prod_{i:v_{i},v_{i+1}\in V_2} \chi(v_{i}v_{i+1})
\end{equation}
where the $y_e$ are edges variables, defined symmetrically: $y_{uv}=y_{vu}$ for all $u,v\in V$. We see that only $k$-walks with exactly $k_1$ vertices from $V_1$ and exactly $k_2$ edges from $V_2\times V_2$ contribute to $[z^k e_{[k_1+k_2]}] \chi(W)$. Equivalently, we assign $\chi(v)=1$ for all $v\in V_2$ and $\chi(uv)=1$ for all $(u,v)\not \in V_2\times V_2$, and then the products can be taken over all vertices and edges of $W$.

To make the notation even more compact, we can now redefine the $\chi(v)$ extensors to be multiplied by a $z$ variable, and the $\chi(uv)$ extensors to be multiplied by $y_{uv}$. Then, the definition is simply:

\begin{equation}
    \chi(W) = \prod_{i=1}^s \chi(v_i) \cdot \prod_{i=1}^{s-1} \chi(v_{i}v_{i+1}).
\end{equation}

Now, instead of computing $\sum_W \chi(W)$ mod $z^{k+1}$ over all the walks, we compute it over what is called \emph{admissible} walks. Admissible walks are walks that do not at any point go from a $V_2$ vertex to a $V_1$ vertex and back to the same $V_2$ vertex. In other words, $W=v_1...v_s$ does not contain any substring of the form $vuv$ where $v\in V_2$ and $u$ in $V_1$. Thus, denoting by $\mathcal{W}$ the set of all walks, our target is the coefficient of $z^k e_{[k_1+k_2]}$ in the expression

\begin{equation}
    \Z := \sum_{\substack{W=(v_1,...,v_s) \in \mathcal{W}\\ W\text{ is admissible}}} y_{v_1} \cdot \chi(W),
\end{equation}
where $y_v$ are vertex variables, and we multiply only by the variable corresponding to the first vertex in the walk. This is done to break the symmetry between the two traversals of the same walk, preventing the whole sum from being 0 mod 2. We note that a walk extensor can be nonzero if it repeats only $V_2$ vertices, but not if it repeats any $V_1$ vertex.
A clever pairing argument described in \cite{narrow-sieves}, using the fact that calculations are over a field of characteristic 2, can be shown to be a nonzero polynomial in the $y$ variables if and only if there exists a $k$-path with exactly $k_1$ vertices from $V_1$ and exactly $k_2$ edges from $V_2\times V_2$. Our goal, then, is to maintain $\Z$ under $\ell$ changes to the graph. In the static case, it is who in \cite{narrow-sieves} that the sum can be computed efficiently with a dynamic programming approach.

Since all $y$ variables are being substituted with scalars, all calculations are done in the ring $\Lambda(\FF^{k_1+k_2}[z] / z^{k+1})$.

\subsubsection{Randomized \texorpdfstring{$O^*(1.66^k)$}{O*(1.66\^k)} sensitivity oracle for \texorpdfstring{$k$}{k}-path detection on undirected graphs}

We now prove \Cref{thm:so-undirected-fast}. We note that, as the current best static algorithm for this problem is upper bounded by $O(1.66^k n^3)$ time, we cannot hope to improve the preprocessing time without improving the upper bound on the static problem too (since then querying with $\ell=0$ updates will output the solution to the static problem). We similarly cannot hope to improve the query time of the sensitivity oracle we present here, as the preprocessing of an empty graph requires at most $O(n^2)$ time, and then querying after insertion of all $|E|\le kn$ edges would result in a better bound.

We will follow a similar idea as we did in \Cref{sec:k-path-query}. One particular detail will present a crucial complication: when stitching together two admissible walks with an edge, the resulting walk does not have to be admissible. This complicates the way in which we compute the stitching, and makes the expressions used more cumbersome.

We describe the operations done in preprocessing time and query time, for any fixed partition $V=V_1\cup V_2$. The preprocessing and query operations are then repeated over the partitions initially chosen at random, that match the number of partitions required.

In preprocessing time, we compute $Q[u,v]$ for any $u,v\in V$ as the sum of all walk extensors over admissible walks from $u$ to $v$, and we similarly let $S[u]$ be the sum of all walk extensors for walks originating at $u$. Here we do not define the dual $F$ vector separately, because $S=F$. We similarly compute a vector $\mathbf{S}$ such that $\mathbf{S}[u]$ is the sum of all walk extensors for walks originating at $u$, multiplied by the vertex variable for the end of the walk. Thus, we have $S[u] = \sum_v Q[u,v]$ while $\mathbf{S}[u] = \sum_v Q[u,v]y_v$. Computing $Q$ (and hence also $S$ and $\mathbf{S}$) can be done in time $2^{k_1+k_2}\poly(k)n^3$ in the same way described in \Cref{sec:narrow-sieves} and similar to the query-time computation below, noting that the number of edges is assumed to be at most $kn+\ell \le (k+1)n$. Therefore, we will not detail it here.

When only adding edges, we compute the sum over the admissible walks in the new graphs with a dynamic-programming approach, as follows. Denote by $A_r[u,v]$ the sum of walk extensors over all admissible walks from vertex $u$ to vertex $v$ in the new graph using exactly $r$ new edges (possibly with repeats), and let $A=\sum_{r=0}^k A_r$. Let $I$ be the set of vertices that are endpoints of one of the new edges. We will consider $A_r[u,v]$ only for $u,v\in I$, thus $A_r$ is a square matrix of shape at most $(2\ell)\times (2\ell)$. Denote by $E^+$ the set of newly-inserted edges, where we insert them \emph{directed}, so $|E^+|=2\ell$. Then, $\Znew - \Z$ can be seen to be composed of
\begin{equation}
    \Znew - \Z = \sum_{r=1}^k Y_r,
\end{equation}
where $Y_r$ is the sum of walk extensors over admissible walks with exactly $r$ new edges (possibly with repeats). Computing $Y_1$ can be done via

\begin{equation}
Y_1 = \sum_{(s,t)\in E^+} \mathbf{S}[s]\chi(st)S[t] .
\end{equation}
(Note that, since $(s,t)$ is a new edge, it was never included in any walk accounted for in $S[s]$ and $S[t]$, and hence the concatenations will always produce an admissible walk.) For $r\ge 2$, $Y_r$ can be computed by singling out the first and last new edges used, resulting in

\begin{equation}
    Y_r = \sum_{(s_1,t_1),(s_2,t_2)\in E^+} \mathbf{S}[s_1]\chi(s_1 t_1)\tilde{A}_{r-2}[t_1,s_2;s_1,t_2]\chi(s_2 t_2)S[t_2] \quad\quad\quad\quad (r\ge 2).
\end{equation}

Here, $\tilde{A}_r[t_1,s_2;s_1,t_2]$ is defined to account for the admissible walks from $t_1$ to $s_2$ that can be used when connecting $t_1$ to $s_1$, and $s_2$ to $t_2$, while using exactly $r$ new edges. That is, walks that would not make the resulting compound walk inadmissible.
Thus,

\begin{equation}
    \sum_{r=2}^k Y_r = \sum_{(s_1,t_1)\in E^+}\sum_{(s_2,t_2)\in E^+} \mathbf{S}[s_1]\chi(s_1 t_1)\tilde{A}[t_1,s_2;s_1,t_2]\chi(s_2 t_2)S[t_2].
\end{equation}

By an inclusion-exclusion argument it is seen that 

\begin{align}
\begin{split}
\tilde{A}[t_1,s_2;s_1,t_2] &= A[t_1,s_2] \\
 &- [t_1\in V_1,s_1\in V_2]\chi(t_1)\chi(s_1t_1)A[s_1,t_1] \\
 &- [s_2\in V_1,t_2\in V_2]\chi(s_2)\chi(s_2t_2)A[s_2,t_2] \\
 &+ [t_1\in V_1,s_1\in V_2,s_2\in V_1,t_2\in V_2]\chi(t_1)\chi(t_1s_1)A[s_1,t_2]\chi(t_2s_2)\chi(s_2).
\end{split}
\end{align}

Since the calculation is done over $\Lambda(\FF^{k_1+k_2}[z, \{y_{uv}\}] / z^{k+1})$ (with scalars substituted for the $y$ variables), each operation requires $2^{k_1+k_2}\poly(k)$ time. Although subtraction is the same as addition in this ring, we keep track of the signs for clarity. Thus, once the values of $A[u,v]$ are computed, only $\ell^2 2^{k_1+k_2}\poly(k)$ time is required to compute $\Znew$.

To compute $A_r[u,v]$ for all $u,v\in V$, we observe that $A_0=Q$, and that by defining $A_{-1}=0$ and conditioning on the first occurrence of a new edge in the walk, we have

\begin{equation}\label{eq:Ar-recursive}
    A_r[u,v] = \sum_{(s,t)\in E^+} Q[u,s] \chi(st) \bigg(A_{r-1}[t, v] - [s\in V_2,t\in V_1] \chi(t)\chi(ts)A_{r-2}[s,v]\bigg),
\end{equation}
where we only compute $A_r[u,v]$ for $u,v\in I$. Note that the right hand side of \Cref{eq:Ar-recursive} only includes $A_r$ and $Q$ in such indices in this case. Thus, we can compute all $A_r$ for $r=0,1,...,k$ in time $\ell^3 2^{k_1+k_2}\poly(k)$. Thus, the whole update is done in time $\ell^3 2^{k_1+k_2}\poly(k)$. Repeating over the random several random partitioning $V=V_1\cup V_2$ results in a $O(\ell^3 1.66^k)$ bound on the query time.

Now suppose we remove one edge from the initial graph.
We first observe the conditions for stitching two walks $W_1$ and $W_2$ such that the resulting walk is admissible. Suppose $W_1$ is a walk ending in $s$, $(s,t)$ is an edge, and $W_2$ is a walk starting at $t$. Then the concatenation, $W$, of these walks, is admissible if and only if $W_1$ and $W_2$ are admissible, with one possible caveat: if $s\in V_2$ and $t\in V_1$, then $W_2$ must not have $s$ as the second element. Likewise, we have the symmetric case of $t\in V_2$ and $s\in V_1$. 

Thus, to compute the sum over admissible walks joined by a removed edge $(s,t)$, we first compute $\mathbf{S}[s]\chi(st)S[t]$ then subtract the resulting forbidden walks, in case they can be formed. This will remove walks that used exactly once the edge $(s,t)$, and in general, a walk that used $(s,t)$ exactly $c$ times will be removed exactly $c$ times.
If $s\in V_2,t\in V_1$ then we must subtract walks that start at some $W_1$ that ends in $s$, then $W_2$ begins with $t$ then $s$. The sum over the walks of this latter type is exactly $S[t]-\chi(t)\chi(st)S[s]$. This is because in this case, walks of the type $t\to s\to W'$ are admissible if and only if $s\to W'$ is admissible. Thus the subtracted walks amount to
\[\bigg(\mathbf{S}[s] - [s\in V_1,t\in V_2]\chi(s)\chi(st)\mathbf{S}[t]\bigg)\chi(st)\bigg(S[t] - [s\in V_2,t\in V_1]\chi(t)\chi(st)S[s]\bigg).\]
Denote $\tilde{S}[s;t] = S[s] - [s\in V_1,t\in V_2]\chi(s)\chi(st)S[t]$, which is, informally, the sum of admissible walks ending at $s$ that should be taken into account when stitching a walk with an edge $(s,t)$. We similarly define $\tilde{\mathbf{S}}[s;t]$. Thus the newly added admissible walks can be more succinctly written as $\tilde{\mathbf{S}}[s;t]\chi(st)\tilde{S}[t;s]$.

Still, we have removed some walks too many times. Only walks using $(s,t)$ exactly once were correctly cancelled, but walks using $(s,t)$ twice were removed twice. We therefore add back $\tilde{\mathbf{S}}[s;t]\chi(st)\tilde{Q}[t,s;s,t]\chi(st)\tilde{S}[t;s]$. Continuing in this way, we now add back walks using $(s,t)$ three times, etc., until we obtain the final result after $k$ iterations\footnote{In fact, it can be seen that traversing the same edge more than twice will always produce a walk extensor evaluating to zero, and thus fewer iterations suffice. However, this observation will not simplify the generalization that follows.}:
\[
\Znew = \Z - \tilde{\mathbf{S}}[s;t]\chi(st)\tilde{S}[t;s] \sum_{i=0}^k \Big(-\tilde{Q}[t,s;s,t]\chi(st)\Big)^i
\]

In general, for both edge insertions and deletions, we can combine the same methods, as follows. We compute the matirx $A_{a,r}[u,v]$ for all $u,v\in I$ and $a,r\ge 0$, $a+r\le k$, which are admissible walks constructed by explicitly placing $a$ added and $r$ removed edges. This notion of ``explicitly'' placing $r$ removed edges may be confusing, and is not equivalent to counting walks using exactly $r$ or at least $r$ removed edges. Instead, it is equivalent to counting exactly $\binom{c}{r}$ times each walk that uses $c$ removed edges. Then more precisely, $A_{a,r}[u,v]$ counts exactly $\binom{c_r}{r}\cdot [c_a=a]$ times each walk that uses exactly $c_r$ removed edges and $c_a$ added edges.

We have $A_{0,0}=Q$, and:

\begin{equation}\label{eq:Ar-full-recursive}
\begin{split}
    A_{a, r}[u,v] &= \sum_{(s,t)\in E^+} Q[u,s] \chi(st) \bigg(A_{a-1,r}[t, v] - [s\in V_2,t\in V_1] \chi(t)\chi(ts)A_{a-2,r}[s,v]\bigg) \\
    &+ \sum_{(s,t)\in E^-} \tilde{Q}[u,s;t] \chi(st) \bigg(A_{a,r-1}[t, v] - [s\in V_2,t\in V_1] \chi(t)\chi(ts)A_{a,r-2}[s,v]\bigg).
\end{split}
\end{equation}

We then compute:
\begin{equation}
\begin{split}
    \Znew &- \Z = \sum_{a+r=1}(-1)^r \sum_{(s,t)\in E^{\pm}}
    \tilde{\mathbf{S}}[s;t]\chi(st)\tilde{S}[t;s] \\
    &+ \sum_{a+r\ge 2}(-1)^r\sum_{(s_1,t_1),(s_2,t_2)\in E^{\pm}} \tilde{\mathbf{S}}[s_1;t_1]\chi(s_1 t_1)\tilde{A}_{r',a'}[t_1,s_2;s_1,t_2]\chi(s_2 t_2)\tilde{S}[t_2;s_2],
\end{split}
\end{equation}
where $r'$ is $r$ minus the number of $E^-$ edges among $(s_1,t_1),(s_2,t_2)$, and $a'$ is $a$ minus the number of $E^+$ edges among them. This concludes the proof of \Cref{thm:so-undirected-fast}.

An argument similar to that used in the proof of \Cref{lem:Z-new} shows that all admissible walks in the new graph with nonzero extensor value are counted exactly once.

\undirectedbipartiteso*

\begin{proof}[Proof sketch]
    This follows in the same way as described in \cite{narrow-sieves} for the static case, by using the $S$ and $T$ as $V_1$ and $V_2$ in the sensitivity oracle for the undirected case, noting that there will be no $V_2\times V_2$ edges and hence we can set $k_2=0$ and $k_1=k/2$.
\end{proof}

\section{Dynamic algorithms} \label{sec:dynamicalgs}
\subsection{\textsc{Exact \texorpdfstring{$k$}{k}-Partial Cover}} \label{sec:exact-partial-cover}
We now present a dynamic algorithm for a different problem, using similar techniques. The problem is a parameterization of the \textsc{Exact Set Cover}. Recall \Cref{def:exact-partial-cover}:
\exactpartialcover*

This is ``exact'' in two senses. The first is that the cover is by disjoint subsets (so every covered element is covered exactly once), and the second is that the size of the covered set is exactly $k$. Both are contrasted with the \textsc{$k$-Partial Cover} problem discussed in \Cref{sec:partial-cover}.

We impose no constraints on the size of $T$ in the above definition, but we note that our techniques can also be adapted for such constraints (for example, deciding the existence under the constraint $|T|\le t$ for some new parameter $t$, or finding the minimum $|T|$ of a solution). 

In the dynamic setting of this problem, subsets $S_i$ are added one by one, and some possibly removed after being previously added. Our result can be stated as follows:

\begin{restatable}{theorem}{exactcovertheorem}\label{thm:exact-set-cover}
	The \textsc{Exact $k$-Partial Cover} problem admits a randomized (one-sided error) dynamic algorithm with $O^*(2^k)$ update time, and a deterministic dynamic algorithm with $O^*(4^k)$ update time. Here $O^*$ hides factors that are polynomial in $k$, $\log n$ and $\log N$.
\end{restatable}

We mention that the running times in \Cref{thm:exact-set-cover} assume we are able to query the size of a set in $O^*(1)$ time. Otherwise, we need to allow $O^*$ to hide factors polynomial in the size of the updated set.

The techniques used are similar to those used by Koutis \cite{koutis-packing} and Koutis and Williams \cite{algebraic-problems}, using the algebra that is isomorphic to $\Lambda(\FF^k)$ for $\FF$ of characteristic 2, which we show how to adapt for the dynamic case, and a deterministic version using $\Lambda(\QQ^{2k})$.

\begin{proof}
	We choose a field $\FF$ of characteristic 2, with size $|\FF| = 2^{\lceil \log_2 N\rceil}$, and compute in $\Lambda(\FF^k)$. Again we have a mapping $\chi:U\to \FF^k$ into different Vandermonde vectors (which is possible due to the size of $\FF$), and we define $\chi(S)=\bigwedge_{a \in S} \chi(a)$, where the order of the product is fixed given $S$ (so for example, after sorting the elements of $S$).
	
	The main observation is as follows: there exists a solution if and only if \[ [e_{[k]}]\prod_{i=1}^n \left(1+y_i \chi(S_i)\right) \neq 0\]
	(the product is in $\Lambda(\FF^k)$, so this is the wedge product) as a polynomial in the $y$ variables. Indeed, upon expanding the product, we get $\sum_{I\subseteq [n]} \prod_{i\in I} y_i \chi(S_i)$. All products of non-pairwise-disjoint sets vanish, and only those whose union is of size $k$ contribute to the coefficient of $e_{[k]}$. All $I\subseteq [n]$ that constitute a solution contribute a different monomial $\prod_{i\in I}y_i$ to the coefficient of $e_{[k]}$, with nonzero coefficient.
	
	Thus, the algorithm is to keep track of the product $\prod_{i=1}^n \left(1+y_i \chi(S_i)\right)$, after randomly assigning values $y_i\in\FF$ to the subsets. It is noted that we must remember the $y_i$ we used on the sets currently in the pool, but need not remember those that were removed.
	
	Now, on updates, we do:
	\begin{itemize}
		\item Adding a set $S$: we choose a random $y_S\in \FF$ and multiply the product by $1+y_S \chi(S)$.
		\item Removing a set $S$: we multiply again product by the same $1+y_S \chi(S)$, where $y_S$ is the same value used at insertion.
	\end{itemize}
	
	Correctness of the increments is evident. For the decrements, we note that $\Lambda(\FF^k)$ is commutative since $\Char \FF=2$, and hence we can assume the maintained product was computed with $S$ being the last inserted set. Then the effect of multiplying by $1+y_S \chi(S)$ again amounts to multiplying in total by $(1+y_S \chi(S))^2 = 1+y_S^2 \chi(S)^2 = 1$, because $\Char \FF=2$ and $\chi(S)^2=0$. Hence, we effectively remove the factor corresponding to $S$.
	
	This finishes the proof for the randomized case. For the deterministic case, we work over $\QQ$ instead of an $\FF$ with characteristic 2, and use the usual trick of lifting to work in $\Lambda (\QQ^{2k})$, using $\bar{\chi}(v) = \binom{\chi(v)}{0}\wedge \binom{0}{\chi(v)}$. Now there is no need for the $y_S$ variables (or equivalently, we set them all to 1), and removing a set now amounts to multiplying by $1-\bar{\chi}(S)$. While here the algebra is not commutative, the factors in the product do commute. Indeed, since all values of $\bar{\chi}$ are spanned by $\{e_i\wedge e_j:i,j\in [2k]\}$, all extensor values involved are seen to have only even-degree monomials, which always commute with any extensor (see~\Cref{sec:extensor-coding}), hence the computation remains valid\footnote{Put differently, all calculations take place in $\bigoplus_{i=0}^k \Lambda^{2i}(\QQ^{2k})$, which is commutative.}. We further note that each product with a term of the form $1\pm\bar{\chi}(S)$ can be done in $O^*(4^k)$ time, because multiplying by $\bar{\chi}(S)$ amounts to iteratively multiplying by each vector composing it, in order, thus reducing to skew multiplications.
\end{proof}

We note that, just as with the \textsc{$k$-Path} problem, we can adapt the deterministic version to an approximate-counting algorithm. We state the corollary for the special case of the \textsc{$m$-Set $k$-Packing} problem in the following subsection.

We are also able to adapt this solution to produce a dynamic algorithm for the problem of determining whether there exists $I\subseteq [n]$ with $\big|\bigcupdot_{i\in I} S_i\big| \ge k$, having randomized update time $O^*(4^k)$ and deterministic $O^*(8^k)$. To that end, we keep track of the product only over sets in the pool whose size is at most $k$ - if any set in the pool has size at least $k$, we announce that a solution exists. Otherwise, we run the algorithm for all $k'$ with $k \le k' \le 2k$, and announce there is a solution if and only if there is a solution to one of these values for $k'$.

\subsection{\textsc{\texorpdfstring{$k$}{k}-Partial Cover}} \label{sec:partial-cover}
We move to a different problem that is a parameterization of the \textsc{Set Cover} problem. Recall \Cref{def:partial-cover}:
\partialcover*

In the dynamic setting of this problem, subsets $S_i$ are added one by one, and some possibly removed after being previously added. Our result can be stated as follows:

\begin{restatable}{theorem}{partialcovertheorem}\label{thm:partial-cover}
The \textsc{$k$-Partial Cover} problem admits a randomized dynamic algorithm with $O^*(2^k)$ update time, and a deterministic dynamic algorithm with $O^*(4^k)$ update time. Here $O^*$ hides factors that are polynomial in $k$, $\log n$ and $\log N$.
\end{restatable}

As with \Cref{thm:exact-set-cover}, we mention that the running times in \Cref{thm:partial-cover} assume we are able to query the size of a set in $O^*(1)$ time. Otherwise, we need to allow $O^*$ to hide factors polynomial in the size of the updated set.

\begin{proof}
The techniques used are similar to those used in Koutis and Williams \cite{algebraic-problems}, which specifically propose the polynomial we use in the randomized case, and which we generalize here as well.

We use the same observation implicit in \cite{algebraic-problems}, that it is possible to cover at least $k$ elements with $t$ sets, if and only if we can \emph{exactly} cover $k$ elements with $t$ pairwise disjoint subsets of the $S_i$s. It follows that there is a solution with $|I|=t$ if and only if $\left(\sum_{i=1}^n \prod_{a\in S_i}(1+x_a)\right)^t$ contains a $k$-multilinear term (that is, a monomial of degree $k$ where all variables appear with exponent at most 1). There is a subtlety here, however, when working over a finite field, as the relevant coefficients can be divisible by the characteristic. By contrast, in other polynomials we used, all the relevant monomials had a coefficient of 1. To deal with this issue, instead of raising the polynomial to the power of $t$, we use different variables in each occurrence.

We introduce the variables $\{y_{a,t}:a\in U, t\in [k]\}$. Now, for any $t\le k$, there exists a solution if and only if
\[ [e_{[k]}]\prod_{j=1}^t \left(\sum_{i=1}^n \prod_{a\in S_i}(1+y_{a,j} \chi(a))\right) \neq 0 \]
as a polynomial in the $y$ variables, and we seek the minimum value of $t$ for which this is true. We note that, if there exists a solution, one must exist with $t\le k$.

To dynamically maintain the answer, we can therefore proceed as follows. First, we maintain $n$ and the number of sets we have of size at least $k$, and as long as this number is positive, we output that the minimal solution is $\min |I| = 1$. We dynamically for each $j\in [k]$ we maintain the value $P_j = \sum_{S:|S| < k}\prod_{a\in S}(1+y_{a,j} \chi(a))$ which runs only over the sets of cardinality less than $k$, and if there are no sets of cardinality at least $k$, we look for the first $t$ such that $[e_{[k]}]\prod_{j=1}^t P_j \neq 0$ (or declare than no such $t$ exists). This can be done with a simple iterative search between the extremes $t=1$ and $t=k$.

For the randomized case, we pick $\Char \FF=2$ as always, then updating $P$ requires $O^*(2^k)$ time, and computing the minimum $t$ requires $O^*(2^k)$ as well (due to complexity of multiplication in characteristic 2, see \Cref{par:general-prod-char-2}). We thus update in time $O^*(2^k)$ in total.

To make this deterministic, we can pick lifts of Vandermonde vectors to work in $\Lambda(\QQ^{2k})$ as always. Updating $P$ requires $O^*(4^k)$ time because we only compute sums and skew-products, but computing the minimal $t$ now requires $O(2^{\omega\cdot 2k / 2}) = O(2^{\omega k})$ field operations.
For this reason, we use a different computation for the deterministic case (which can also be used for the randomized case). To that end, we introduce a single new variable $z$, and consider the polynomial
\begin{equation}
    P(z) = \prod_S \left(1 + z\left[\prod_{a\in S} (1+\bar{\chi}(a)) - 1\right] \right).
\end{equation}
We note that $[e_{[2k]}z^t]P\neq 0$ (that is, the coefficient of $e_{[2k]}$ in the coefficient of $z^t$ in $P$) if and only if we can cover $k$ elements with exactly $t$ disjoint nonempty subsets of the $S$'s. Thus, the lowest $t$ for which this is nonzero is the answer we seek. Because the minimal solution (if one exists) always has $t\le k$, it is enough to consider $P(z)$ modulo $z^{k+1}$. However, we see that $\deg P(z) \le k$ because any product of more than $k$ extensors of the form $\bar\chi(a)$ vanishes. It follows that the size of $P$ is not too high to handle dynamically.

To update the polynomial $P(z)$ with a new set $S$ of size $|S|<k$, we can multiply $P(z)$ by the corresponding factor $1 + z\left[\prod_{a\in S} (1+\bar{\chi}(a)) - 1\right]$ with only additions and skew multiplications, thus requiring only $O^*(4^k)$ field operations. Indeed, the product of $P(z)$ with the new factor is $(1-z)P(z) + P(z)\prod_{a\in S} (1+\bar{\chi}(a))$, which can be performed by separately computing $(1-z)P(z)$ and $P(z)\prod_{a\in S} (1+\bar{\chi}(a))$, where the latter can be computed by repeated skew-multiplications.

It remains to describe how we can remove a set with the same time complexity. For that, we note that $1 + z\left[\prod_{a\in S} (1+\bar{\chi}(a)) - 1\right]$ has an inverse in $\Lambda(\QQ^{2k})[z]$. Indeed, denote $X = z\left[\prod_{a\in S} (1+\bar{\chi}(a)) - 1\right]$. Then $X$ has no constant extensor term (that is, $[e_{\emptyset}]X$ is the zero polynomial), and hence $X^{k+1}=0$ (as all monomials of degree more than $2k$ equal zero, and $X$ itself contains only monomials of degree at least 2). Therefore, \[(1+X)\left(1-X+X^2-X^3+...+(-X)^k\right) = 1 - (-X)^{k+1} = 1.\]
Hence, we can multiply $P$ by $(1-X+X^2-X^3+...+(-X)^k)$ to remove the factor corresponding to $S$. Since multiplying by $X$ requires only a $\poly(k)$ additions and skew multiplications, we are able to multiply by $(1+X)^{-1}$ in $O^*(4^k)$ time, as promised.
\end{proof}

\subsection{\textsc{\texorpdfstring{$m$}{m}-Set \texorpdfstring{$k$}{k}-Packing}} \label{sec:set-packing}

Recall \Cref{def:m-set-k-packing}:
\setpacking*

This problem is studied, for example, in \cite{koutis-packing}. This is then a special case of the \textsc{Exact $mk$-Partial Cover} discussed in \Cref{sec:exact-partial-cover} (with the same input universe $U$ and subsets $S_i$), and hence we obtain:

\begin{restatable}{theorem}{setpackingtheorem}\label{thm:m-set-k-packing}
The \textsc{$m$-Set $k$-Packing} problem admits a randomized (one-sided error) dynamic algorithm with $O^*(2^{mk})$ update time, and a deterministic dynamic algorithm with $O^*(4^{mk})$ update time. Here $O^*$ hides factors that are polynomial in $m,k,\log n$ and $\log N$.
\end{restatable}

Applying the randomized dynamic algorithm to the static problem recovers the running time $2^{mk}\poly(mk) n$ given in Koutis~\cite{koutis-packing} for a randomized algorithm.

As described, we can also adapt the methods to approximately count the number of solutions with a dynamic algorithm.

\begin{corollary} \label{cor:m-set-k-packing-count}
	The \textsc{$m$-Set $k$-Packing} approximate counting problem admits a randomized dynamic algorithm with $O^*\left(\frac{1}{\epsilon^2}4^{mk}\right)$ update time. Here $O^*$ hides factors that are polynomial in $m,k$, $\log n$ and $\log N$. The algorithm outputs an estimate to the number of solutions, that is within relative error $\epsilon$ with probability $> 99\%$.
\end{corollary}

This is contrasted with the static \emph{exact} counting algorithm presented in \cite{algebraic-problems}, running in time $n^{\lceil mk/2\rceil}\poly(n)$.

\subsection{\textsc{\texorpdfstring{$t$}{t}-Dominating Set}} \label{sec:t-dominating}

We next focus on a parameterization of a different \textsf{NP}-Hard problem: the \textsc{Dominating Set}. Recall \Cref{def:t-dominating}:
\tdominating*

This is contrasted with a different parameterization of the \textsc{Dominating Set} problem that the reader might be familiar with, that is known to be W[2]-complete \cite{parameterized-algorithms-book}. Namely, the problem of deciding whether a graph has an $n$-dominating set (that is, a set of vertices dominating the whole graph) of size at most $k$. The \textsc{$t$-Dominating Set} problem, in contrast, is Fixed-Parameter Tractable.

This is then a special case of the \textsc{$k$-Partial Cover} discussed in \Cref{sec:partial-cover} (for $k=t$), where the universe size is a bound on the number of vertices, the sets are $\forall v\in V(G): S_v := \{v\}\cup N(v)$.

We are therefore able to support vertex updates (additions and removals) and also edge updates (by removing the two affected sets and inserting the updated ones) in a dynamic setting. We have thus shown:

\begin{restatable}{theorem}{tdominatingtheorem}\label{thm:t-dominating}
	The \textsc{$t$-Dominating Set} problem admits a randomized (one-sided error) dynamic algorithm with $O^*(2^t)$ update time, and a deterministic dynamic algorithm with $O^*(4^t)$ update time. Here $O^*$ hides factors that are polynomial in $t$, $\log n$ and $\log N$.
\end{restatable}

\subsection{\texorpdfstring{$d$}{d}-Dimensional \texorpdfstring{$k$}{k}-Matching} \label{sec:dimensionalmatching}
Another natural question that is solved dynamically with the same methods is the dynamic $k$-Matching problem, which asks to maintain whether or not a (not necessarily bipartite) undirected graph has a matching of at least $k$ pairs of vertices.

This is reducible to the \textsc{$2$-Set $k$-Packing} problem, by mapping a edge $e_i=\{u_i,v_i\}$ to the set $S_i=\{u_i, v_i\}$, and hence we inherit the running times presented in \Cref{thm:m-set-k-packing} and \Cref{cor:m-set-k-packing-count} with $m=2$.

For bipartite graphs, the problem is generalized by the well-studied \textsc{$d$-Dimensional $k$-Matching} problem. For example, see \cite{algebraic-problems}. Recall \Cref{def:d-dim-k-matching}:
\dimensionalmatching*

This generalization is also reducible to \textsc{$m$-Set $k$-Packing}, allowing for $O^*(2^{dk})$ randomized and $O^*(4^{dk})$ deterministic dynamic solutions, but borrowing ideas of \cite{algebraic-problems} we can produce faster dynamic solutions, that for the randomized static case match those of \cite{algebraic-problems}.

\begin{restatable}{theorem}{kmatchingtheorem}\label{thm:k-matching}
	The \textsc{$d$-Dimensional $k$-Matching} problem admits a randomized (one-sided error) dynamic algorithm with $O^*(2^{(d-1)k})$ update time, and a deterministic dynamic algorithm with $O^*(4^{(d-1)k})$ update time. $d$ is regarded as constant, and $O^*$ hides factors that are polynomial in $k$, $\log n$ and $\log N$.
\end{restatable}

 We note that \cite{narrow-sieves} presents a randomized algorithm for the static problem in time $O^*(2^{(d-2)k}\poly(n))$.

\begin{proof}
We build on the same polynomial constructed in \cite{algebraic-problems} for this problem, and begin by discussing a randomized dynamic algorithm. Each element $a\in \cupdot_{i=2}^d U_i$ is mapped to an extensor, as usual, but now with a different dimension $\chi:\cupdot_{i=2}^d U_i\to\FF^{(d-1)k}$, where $\FF$ is a field of characteristic 2. Now for each $t=(t[1],...,t[d])\in T$ we multiply all extensors given only to the entries $2$ through $d$, as $M_t = \prod_{i=2}^d \chi(t[i])$. We also introduce tuple variables, $y_t$ for each $t\in T$. For any $a\in U_1$ we define $Q_a = \sum_{t:t[1]=a} M_t y_t$.
Then, introducing a new variable $z$, we consider \begin{equation}
    P(z):=\prod_{a\in U_1} \left(1 + z Q_a \right)=\prod_{a\in U_1} \left(1 + z \sum_{t:t[1]=a} M_t y_t \right).
\end{equation}
We note that the coefficient of $z^k e_{[(d-1)k]}$ in $P(z)$ is a nonzero polynomial in the $y$ variables if and only if there exists a solution. Furthermore, $\deg P(z) \le k$ because any term contributing to higher powers of $z$ vanish, as products of more than $(d-1)k$ vectors.

Now we can proceed similarly to what we did in the \textsc{$k$-Partial Cover} problem (see \Cref{thm:partial-cover} in \Cref{sec:partial-cover}). Here, we will maintain each $Q_a$. An insertion or removal of a tuple $t$ only changes one of the $Q$ values, namely $Q_{t[1]}$, which can be done in time $O^*(2^{(d-1)k})$.

Then, we can update $P(z)$ as follows. Mark the previous value of $Q_{t[1]}$ as $Q'_{t[1]}$ (and the new value is denoted $Q_{t[1]}$). Then we first cancel the previous factor $1+zQ'_{t[1]}$ from $P(z)$, then add the new factor $1+zQ_{t[1]}$. Adding the new factor amounts to multiplying two extensors (or, in fact, performing $\poly(k)$ extensors multiplications, because we are now working with polynomials) over characteristic 2, which requires $O^*(2^{(d-1)k})$.

As for removing the old factor, we note that since over a field of characteristic 2 we have $(a+b)^2=a^2+b^2$ for any two extensors $a$ and $b$, and because squares of vectors vanish, $(1+zQ'_{t[1]})^2 = 1^2 + z^2 (Q'_{t[1]})^2 = 1$, since $Q'_{t[1]}$ does not have any constant extensor term. Thus, removing the factor $1+zQ'_{t[1]}$ amounts to multiplying by the same thing, which can be done in $O^*(2^{(d-1)k})$ time.

This finishes the proof for the randomized case.

For the deterministic case, we make the usual changes. Put $y_t=1$ for all $t$, double the underlying field dimension from $(d-1)k$ to $2(d-1)k$, change $\FF$ to $\QQ$, and lift the extensors $\chi:\cupdot_{i=2}^d U_i\to\QQ^{(d-1)k}$ to $\bar{\chi}:\cupdot_{i=2}^d U_i\to\QQ^{2(d-1)k}$.

We now follow the same technique employed in the proof of \Cref{thm:partial-cover}. Changing $Q_{t[1]}$ for the updated $t$ now requires $O^*(4^{(d-1)k})$ time. Multiplying by the new factor $1+zQ_{t[1]}$ seems harder, but we note that all terms in $Q_{t[1]}$ are products of exactly $d-1$ of the $2(d-1)k$ basis elements $e_i$ of $\QQ^{2(d-1)k}$. Thus, there are at most $\binom{2(d-1)k}{d-1} = \poly(k)$ such terms (recall we assume $d$ is constant), and we can perform a multiplication of a general extensor by $Q_{t[1]}$ by simply expanding the product, requiring $O^*(4^{(d-1)k}\poly(k)) = O^*(4^{(d-1)k})$ time.

Removing the factor $1+zQ'_{t[1]}$ is done by noting that $(Q'_{t[1]})^{k+1}=0$, so we can compute its inverse as $(1+zQ'_{t[1]})^{-1}=\sum_{i=0}^{k} (-zQ'_{t[1]})^i$. We do not compute it explicitly, but rather can effectively multiply by it by repeatedly multiplying by $Q'_{t[1]}$, which we already described how to perform in $O^*(4^{(d-1)k})$ time. Since the number of such operations required is only $\poly(k)$, we see that all computations require $O^*(4^{(d-1)k})$ time in total.
\end{proof}

\section{Adding bounding constraints}\label{sec:constraints}
To further demonstrate the robustness of the exterior algebra-based techniques, we show how we can add bounding constraints to the problems discussed, and still employ the same methods. These ideas are also directly applicable to generalizations of the \textsc{$k$-MlD} ($k$-Multilinear Detection) problem. In this problem, one is interested in deciding whether an arithmetic circuit computes a polynomial in several variables which, when expanded into a sum of monomials, contains has at least one multilinear monomial (that is, a monomial where all variables appear with exponent at most 1) of degree $k$ with a nonzero coefficient.

The basic idea we generalize was given in Koutis~\cite{koutis-constraints} for the constrained multilinear detection problem (\textsc{$k$-CMlD}) over an fields of characteristic 2, motivated by the \textsc{Graph Motif} problem.

\paragraph*{Unique colors}
The idea presented in Koutis~\cite{koutis-constraints} is phrased as an algebraic algorithm using an algebraic structure isomorphic to $\Lambda(\FF^{k})$ for $\FF$ of characteristic 2. We repeat the idea here for the exterior algebra, generalizing the technique. Given an arithmetic circuit on $n$ input variables, we wish to detect whether it contains a degree-$k$ multilinear term with some constraints. Specifically, each variable is also attached a predetermined \emph{color}, and we wish to detect whether there is a degree-$k$ multilinear term in which each color $c$ can only appear a certain $\mu(c)\in\NN$ amount of times.

The solution proposed is to randomly select subspaces of $\FF^k$, such that for each color $c$ we have a different subspace $S_c\subseteq \FF^k$ of dimension $\dim S_c = \mu(c)$ (when $\FF$ is a finite field, this can be achieve by randomly sampling vectors in $\FF^k$ until $\mu(c)$ of them are independent, and then $S_c$ is their span). Then, the vector $\chi(x_i)\in\FF^k$ assigned to the $i$th variable is chosen as a random vector in $S_c$. 

The immediate effect is that any monomial containing more than $\mu(c)$ variables of the color $c$ vanishes. This can be seen by writing each participating $\chi(x_i)$ with the color $c$ as a linear combination of some basis of $S_c$, and noting that upon expanding the product $\bigwedge_i \chi(x_i)$, each term will necessarily repeat a basis element, and therefore vanish. In the other direction, \cite{koutis-constraints} analyzes the probability that degree-$k$ multilinear monomial that meets the constraints does not vanish, when $\FF=\FF_2$ is the field of size 2. The probability can be seen to grow significantly with the size of $\FF$.

\begin{remark}
In cases where we will be interested in $\FF = \QQ$, it is unclear how to randomly pick vectors. However, in this cases we are usually interested in a deterministic algorithm, and indeed we can pick the subspaces deterministically, as spans of Vandermonde vectors, and choose coefficients which are also components of Vandermonde vectors. As an example, see \Cref{constrained-k-path}.
\end{remark}

\paragraph*{Multiple colors}
While each variable receives a single color in the above discussion, one might encounter problems in which a variable has multiple colors, and all of their respective constraints must be fulfilled\footnote{This is contrasted with a different notion of 'multiple colors' sometimes used in the literature. Namely, having each variable 'choose' one of the assigned color and discard the rest.}. Put precisely, we have color space $[C]$, variables $x_i$ for $i\in [n]$, color mapping $\eta: [n] \to 2^{[C]}$, and bounding constraints $\mu:[C]\to \NN$. We wish to decide whether there is a degree-$k$ multilinear monomial $\prod_{j=1}^k x_{i_j}$ where for any color $c$ it holds that $|\{j:c\in \eta(i_j)\}|\le \mu(c)$.

We are also able to deal with these kinds of constraints, in the following way. First, for each color $c\in [C]$ we randomly draw a subspace $S_c$ as before. Then, each variable with this color, i.e., $x_i$ with $c\in\eta(i)$, is given a random vector from $S_c$, say $v^{(c)}_i$. Then, the extensor value $\chi(x_i)$ for each variable is the wedge product (in arbitrary order), $\chi(x_i)=\bigwedge_{c\in\eta(i)} v^{(c)}_i$. While this can be seen to take care of nullifying any monomial that does not meet the constraints, we now run into three immediate problems:
\begin{enumerate}
    \item A monomial might vanish simply because it  has a too-high extensor degree. That is, it might be a product of more than $k$ vectors in $\FF^k$. Therefore, we need to increase the dimension of the underlying vector space, to correctly bound the largest extensor degree we expect in a solution. Let $d$ be a bound, and we continue compute $\Lambda(\FF^d)$.
    \item It might not always be clear which coefficient we should look at in the result. The coefficient of $e_{[d]}$, for example, can only be reached with a multilinear monomial who happens to have the correct sum of extensor degrees. To correctly catch all degree-$k$ multilinear monomials, we introduce a new variable, $z$. We then compute let each $\chi(x_i)$ be a multiple of $z$:
    \[ \chi(x_i)=z\cdot \bigwedge_{c\in\eta(i)} v^{(c)}_i\]
    We can now work in $\Lambda(\FF^d)[z]\operatorname{ mod }z^{k+1}$ instead of $\Lambda(\FF^d)$, and inspect the coefficient of $z^k$ in the result. Since we are now working with polynomials, operations will require more time. However, this is seen to only incur a polynomial factor in $k$. A different solution is to compute the circuit multiple times with several values of $z\in\FF$, then interpolate.\footnote{It is noted that Lagrange interpolation works as-is in the exterior algebra, if the values used for $z$ are all scalars.}
    \item In general $\Lambda(\FF^k)$ is noncommutative, so substituting extensors for variables might not reproduce the expected monomials (for example, if the arithmetic circuit computes $xy - yx = 0$, with extensors this could equal $2xy$ instead). This is not always an issue, but when it is, we can make sure to wedge an even number of vectors for each variable. Now calculations take place in $\bigoplus_i \Lambda^{2i}(\FF^k)\subseteq \Lambda(\FF^k)$, which \emph{is} commutative.
\end{enumerate}

\paragraph*{Repeating variables}
In cases where $\Char \FF$ is not too small, we can also allow for some of the variables to repeat at most some predetermined number of times, counting their colors with the same multiplicity, at the expense of increasing the dimension. For ease of presentation, we describe here only a simple example. Suppose $x_1$ is allowed to appear with multiplicity at most 3. That is, we potentially allow terms of the form $x_1^3x_2x_5$ to count as a solution. We use the folklore trick of replacing $x_1$ by a sum of three other variables, $x_1=x_{1,1}+x_{1,2}+x_{1,3}$ obtaining several monomials upon expanding, one of which is $3! x_{1,1}x_{1,2}x_{1,3}x_2x_5$, while the others are not multilinear in the new variables. We thus need to have $\Char \FF\notin\{2,3\}$, for otherwise the new multilinear term vanishes. But for this to also hold when $\chi(x_1)$ is replaced with $\chi(x_{1,1})+\chi(x_{1,2})+\chi(x_{1,3})$, we must require the $\chi(x_{1,j})$ extensors to commute, or equivalently to have an even extensor degree. If this is not the case, then, we multiply each of the $\chi(x_{1,j})$ extensors by an additional random vector.

We end the discussion on bounding constraints with two examples.

\begin{example}[\textsc{$k$-Walk} repeating at most once]
Given a directed graph $G$ on $n$ vertices, decide whether it contains a $k$-walk with at most one repeating vertex. That is, a walk of length $k$ on at least $k-1$ vertices. We describe an efficient deterministic algorithm for this problem.

To this end, we create two copies of $G$, denoted $G_1$ and $G_2$, and let $G'$ be the union of $G_1$ and $G_2$. For each $u\in V(G)$ denote by $u_1\in V(G_1), u_2\in V(G_2)$ their corresponding copies in $G_1$ and $G_2$. Then, for each $(u,v)\in E(G)$, add the edges $(u_1,v_2)$ and $(u_2,v_1)$ to $G'$. We then choose $n+1$ different Vandermonde vectors in $\QQ^k$, and lift them to extensors of degree 2 in $\Lambda(\QQ^{2k})$ in the usual way. Each vertex in $G_1$ will receive a different such extensor, and the remaining one is given to all vertices in $G_2$. This effectively constraints the a walk to pass in $G_2$ at most once. It is seen that the sum of all walk extensors of walks in $G'$ has a nonzero coefficient of $e_{[2k]}$ if and only if there exists a walk with the determined constraints, and hence we are able to statically detect this in time $4^k \poly(k) n^2$, as well as with a sensitivity oracle, by using the sensitivity oracle designed in \Cref{sec:k-path-so}.
\end{example}

\begin{example}[Constrained \textsc{$k$-Path}] \label{constrained-k-path}
Given a directed graph $G$ on $n$ vertices, two (possibly intersecting) subsets $V_1,V_2\subseteq V$, and two numbers $\mu_1,\mu_2\in\NN$, decide whether $G$ contains a $k$-path that contains at most $\mu_1$ vertices of $V_1$ and at most $\mu_2$ vertices of $V_2$. We describe an efficient deterministic algorithm for this problem as well.

We pick $\mu_1+\mu_2$ different Vandermonde vectors $u_1,...,u_{\mu_1}$, $w_1,...,w_{\mu_2}$ in $\QQ^d$ (for some $d$ chosen later) with positive entries. For the $i$th vertex in $V_1$ we give the vector $\sum_{t=1}^{\mu_1} i^t u_t$ (this construction ensures that any $\mu_1$ of these vectors are linearly independent). We do the same thing to $V_2$, with $\mu_2$ and the $w$ vectors. For all vertices not in $V_1\cup V_2$, we give other Vandermonde vectors.

Now, as we always do in the deterministic algorithms, we lift the vectors to degree-2 extensors in $\Lambda(\FF^{2d})$.
Any vector in the intersection $V_1\cap V_2$ was given two vectors, which were lifted to degree-2 extensors, and we wedge them to obtain a single degree-4 extensor.

Define the \emph{degree} of a walk to be the sum of extensor degrees of the extensors given to the vertices in it.
It can be seen that, as long as $2d$ is at least the degree of a $k$-path that meets the constraints, the sum of walk-extensors of $k$-walks is nonzero if and only if there exists a solution\footnote{For example, by noting that when a solution exists, there are vectors we can wedge to the result and obtain a nonzero coefficient for $e_{[2d]}$.}. Thus, it is enough to take $d=k+\min\{k, |V_1\cap V_2|\}$. Since calculations take place in $\Lambda(\FF^{2d})$, we obtain a deterministic running time of $4^{k+\min\{k, |V_1\cap V_2|\}}\poly(k) n^2$ for this problem, as well as a deterministic sensitivity oracle, by using the sensitivity oracle designed in \Cref{sec:k-path-so}.
\end{example}

\section{Discussion on fixed-parameter complexity classes} \label{sec:dynamic-vs-sensitivity}

We discuss the relationship between the following definitions.

\begin{definition}[FPT] \label{def:fpt}
A parameterized problem is in \textsf{FPT} if it is decidable in time $f(k)\poly(n)$ for a computable function $f$.
\end{definition}

\begin{definition}[FPD] \label{def:fpd}
A parameterized problem is in \textsf{FPD} (Fixed-Parameter Dynamic) if there is a dynamic algorithm for it requiring $f(k)\poly(n)$ preprocessing time and $g(k)n^{o(1)}$ update time, for computable functions $f$ and $g$.
\end{definition}

\begin{definition}[FPSO] \label{def:fpso}
A parameterized problem is in \textsf{FPSO} (Fixed-Parameter Sensitivity Oracle) if there is a sensitivity oracle for it requiring $f(k)\poly(n)$ preprocessing time and $\poly(\ell) g(k)n^{o(1)}$ update time, for a computable function $g$.
\end{definition}

We note that $\textsf{FPD} \subseteq \textsf{FPSO} \subseteq \textsf{FPT}$.
We can show that the inclusions are strict, at least under plausible hardness conjectures. For example, \cite{dynamic-parameterized} shows that \textsc{$k$-Path} in directed graphs does not admit a dynamic algorithm under hardness conjectures (which are shown in this paper to be in $\textsf{FPSO}$). \cite{so-hardness} shows that the \textsc{\#SSR} problem does not have an efficient sensitivity oracle, assuming $\text{SETH}$, which in turn can be used to show that, assuming $\text{SETH}$, there exists an \textsf{FPT} problem that is not in \textsf{FPSO}.

We further note that many problems shown in \cite{dynamic-parameterized} to not be in \textsf{FPD} can be shown to be in \textsf{FPSO}. As a few examples:
\begin{enumerate}
    \item \textsc{Triangle Detection}, which is the problem of detecting whether a graph contains a triangle, can be solved efficiently by a sensitivity oracle by precomputing the square of the graph's adjacency matrix and computing the number of triangles in time $O(n^\omega)$, then updating the number of triangles in time $O(\ell^\omega)$ ($O(\ell)$ for triangles that use a single updated edge, $O(\ell^2)$ for triangles using two updated edges, and $O(\ell^\omega)$ for those using three updated edges).
    \item Incremental \textsc{st-Reachability}, which is the problem of deciding whether two predetermined vertices $s$ and $t$ are connected in a directed graph, only allowing incremental updates, can be solved efficiently by precomputing reachability between any two vertices (for example by running BFS from all vertices in time $\poly(n)$), then using dynamic programming to answer a query in time $\poly(\ell)$ by updating the connectivity information only on the $\le 2\ell+2$ vertices that are either $s,t$ or are part of any inserted edge
    \item \textsc{3SUM}, which is the problem of deciding whether there are 3 elements in a list that sum to 0. Here it is possible to precompute the sums of all pairs in time $O(n^2)$ (counting multiplicities), and the number of triples whose sum is 0. Then when adding or removing $\ell$ numbers it is possible to compute in $\poly(\ell)$ time the number of new solutions and the number of previous solutions that should be removed, and checking if the remaining number of solutions is nonzero.
    \item \textsc{$k$-Layered Reachability Oracle} (\textsc{$k$-LRO}), the problem of deciding whether two vertices $u,v$ are connected in a directed $k$-layered graph (that is, a graph whose vertices are partitioned into $k$ parts, with edges only going from one part to the next), also has an efficient sensitivity oracle, and in fact can be seen to be equivalent to the directed \textsc{$k$-Path} problem. In particular, the same \textsc{$k$-Path} sensitivity oracle devised in \Cref{sec:k-path-so} can be used here without any changes.
\end{enumerate} 

\section{Open questions}

We briefly discuss a few natural questions that arise.

\begin{enumerate}
	\item Is there a fully-dynamic \textsc{$k$-Path} detection algorithm on undirected graphs with $f(k)n^{O(1)}$ preprocessing time for some computable function $f$, and update time $2^{O(k)}n^{o(1)}$?
	
	This would beat the dynamic algorithm proposed in \cite{dynamic-parameterized} that is a simple adaptation of the original color-coding idea \cite{color-coding}, while to the best of our knowledge, no improvement on this original color-coding idea is known to transfer to the dynamic setting. While our work shows a better dependency on $k$, we do not reach a truly dynamic algorithm, but rather only a sensitivity oracle.
	
	We note that it is unlikely that such an an algorithm exists for directed graphs, due to the conditional lower bound presented in \cite{dynamic-parameterized}.
	
	\item Is there an efficient sensitivity oracle for the \textsc{$k$-Tree} problem?
	
	The \textsc{$k$-Tree} problem is as follows: given an undirected graph $G$ on $n$ vertices and a tree $T$ on $k$ vertices, determine whether $T$ is a (not necessarily induced) subgraph of $G$. This problem is known to be \textsf{FPT}, and for example is shown in \cite{algebraic-problems} to be solvable in time $2^k \poly(k) \poly(n)$, with techniques similar to those applied for the \textsc{$k$-Path} problem. However, it is unclear how to adapt the techniques here for an efficient sensitivity oracle. We conjecture that there is, in fact, no efficient sensitivity oracle for this problem. More specifically, we conjecture this for any \textsc{$k$-Tree} formed by connecting a single vertex to the beginning of $\Theta(\sqrt{k})$ paths of length $\Theta(\sqrt{k})$. We note that the techniques of \cite{so-decremental} for a decremental sensitivity oracle work for any \textsc{$k$-Tree} just as well as they do for the \textsc{$k$-Path} problem.
	\item Can other techniques used to solve the static versions of the problems discussed in this paper, or other parameterized problems, be used to design faster dynamic algorithms and sensitivity oracles?
\end{enumerate}	

\section*{Acknowledgements}
We thank Jan van den Brand for telling us about the reference~\cite{algebraic-distance-oracles}, as well as Shyan Akmal, Cornelius Brand, and anonymous reviewers for helpful comments. 
This research was supported in part by a grant from the Simons Foundation (Grant Number 825870 JA), by NSF grant CCF-2107187, by a grant from the Columbia-IBM center for Blockchain and Data Transparency, by JPMorgan Chase \& Co., and by LexisNexis Risk Solutions. Any views or opinions expressed herein are solely those of the authors listed.

\bibliography{main}

\newcommand{\etalchar}[1]{$^{#1}$}
\begin{thebibliography}{HLNW17}

\bibitem[AMW20]{dynamic-parameterized}
Josh Alman, Matthias Mnich, and Virginia~Vassilevska Williams.
\newblock Dynamic parameterized problems and algorithms.
\newblock {\em ACM Trans. Algorithms}, 16(4), 7 2020.
\newblock \href {https://doi.org/10.1145/3395037} {\path{doi:10.1145/3395037}}.

\bibitem[AW14]{AW-lbs}
Amir Abboud and Virginia~Vassilevska Williams.
\newblock Popular conjectures imply strong lower bounds for dynamic problems.
\newblock In {\em 2014 IEEE 55th Annual Symposium on Foundations of Computer
  Science}, pages 434--443. IEEE, 2014.

\bibitem[AW21]{matrix-mult-exponent}
Josh Alman and Virginia~Vassilevska Williams.
\newblock A refined laser method and faster matrix multiplication.
\newblock In {\em Proceedings of the Thirty-Second Annual ACM-SIAM Symposium on
  Discrete Algorithms}, SODA '21, page 522–539, USA, 2021. Society for
  Industrial and Applied Mathematics.
\newblock URL: \url{https://dl.acm.org/doi/10.5555/3458064.3458096}.

\bibitem[AYZ95]{color-coding}
Noga Alon, Raphael Yuster, and Uri Zwick.
\newblock Color-coding.
\newblock {\em J. ACM}, 42(4):844–856, 7 1995.
\newblock \href {https://doi.org/10.1145/210332.210337}
  {\path{doi:10.1145/210332.210337}}.

\bibitem[BCC{\etalchar{+}}22]{so-decremental}
Davide Bil\`{o}, Katrin Casel, Keerti Choudhary, Sarel Cohen, Tobias Friedrich,
  J.A.~Gregor Lagodzinski, Martin Schirneck, and Simon Wietheger.
\newblock {Fixed-Parameter Sensitivity Oracles}.
\newblock In {\em 13th Innovations in Theoretical Computer Science Conference
  (ITCS 2022)}, volume 215 of {\em Leibniz International Proceedings in
  Informatics (LIPIcs)}, pages 23:1--23:18, 2022.
\newblock \href {https://doi.org/10.4230/LIPIcs.ITCS.2022.23}
  {\path{doi:10.4230/LIPIcs.ITCS.2022.23}}.

\bibitem[BDH18]{extensor-coding}
Cornelius Brand, Holger Dell, and Thore Husfeldt.
\newblock Extensor-coding.
\newblock In {\em Proceedings of the 50th Annual ACM SIGACT Symposium on Theory
  of Computing}, STOC 2018, page 151–164, New York, NY, USA, 2018.
  Association for Computing Machinery.
\newblock \href {https://doi.org/10.1145/3188745.3188902}
  {\path{doi:10.1145/3188745.3188902}}.

\bibitem[BHKK07]{fast-subset-convolution}
Andreas Bj\"{o}rklund, Thore Husfeldt, Petteri Kaski, and Mikko Koivisto.
\newblock Fourier meets m\"{o}bius: Fast subset convolution.
\newblock In {\em Proceedings of the Thirty-Ninth Annual ACM Symposium on
  Theory of Computing}, STOC '07, page 67–74, New York, NY, USA, 2007.
  Association for Computing Machinery.
\newblock \href {https://doi.org/10.1145/1250790.1250801}
  {\path{doi:10.1145/1250790.1250801}}.

\bibitem[BHKK17]{narrow-sieves}
Andreas Björklund, Thore Husfeldt, Petteri Kaski, and Mikko Koivisto.
\newblock Narrow sieves for parameterized paths and packings.
\newblock {\em Journal of Computer and System Sciences}, 87:119--139, 2017.
\newblock URL:
  \url{https://www.sciencedirect.com/science/article/pii/S0022000017300314},
  \href {https://doi.org/https://doi.org/10.1016/j.jcss.2017.03.003}
  {\path{doi:https://doi.org/10.1016/j.jcss.2017.03.003}}.

\bibitem[BHRT21]{dynamic-kernels}
Max Bannach, Zacharias Heinrich, R\"{u}diger Reischuk, and Till Tantau.
\newblock {Dynamic Kernels for Hitting Sets and Set Packing}.
\newblock In Petr~A. Golovach and Meirav Zehavi, editors, {\em 16th
  International Symposium on Parameterized and Exact Computation (IPEC 2021)},
  volume 214 of {\em Leibniz International Proceedings in Informatics
  (LIPIcs)}, pages 7:1--7:18, Dagstuhl, Germany, 2021. Schloss Dagstuhl --
  Leibniz-Zentrum f{\"u}r Informatik.
\newblock URL: \url{https://drops.dagstuhl.de/opus/volltexte/2021/15390}, \href
  {https://doi.org/10.4230/LIPIcs.IPEC.2021.7}
  {\path{doi:10.4230/LIPIcs.IPEC.2021.7}}.

\bibitem[Bra19]{patching-colors}
Cornelius Brand.
\newblock {Patching Colors with Tensors}.
\newblock In {\em 27th Annual European Symposium on Algorithms (ESA 2019)},
  volume 144 of {\em Leibniz International Proceedings in Informatics
  (LIPIcs)}, pages 25:1--25:16, 2019.
\newblock \href {https://doi.org/10.4230/LIPIcs.ESA.2019.25}
  {\path{doi:10.4230/LIPIcs.ESA.2019.25}}.

\bibitem[Bra22]{fast-deterministic-skew-product}
Cornelius Brand.
\newblock Discriminantal subset convolution: Refining exterior-algebraic
  methods for parameterized algorithms.
\newblock {\em Journal of Computer and System Sciences}, 129:62--71, 2022.
\newblock URL:
  \url{https://www.sciencedirect.com/science/article/pii/S0022000022000423},
  \href {https://doi.org/https://doi.org/10.1016/j.jcss.2022.05.004}
  {\path{doi:https://doi.org/10.1016/j.jcss.2022.05.004}}.

\bibitem[CCD{\etalchar{+}}21]{dynamic-elim-forests}
Jiehua Chen, Wojciech Czerwi{\'n}ski, Yann Disser, Andreas~Emil Feldmann, Danny
  Hermelin, Wojciech Nadara, Marcin Pilipczuk, Micha{\l} Pilipczuk, Manuel
  Sorge, Bart{\l}omiej Wr{\'o}blewski, et~al.
\newblock Efficient fully dynamic elimination forests with applications to
  detecting long paths and cycles.
\newblock In {\em Proceedings of the 2021 ACM-SIAM Symposium on Discrete
  Algorithms (SODA)}, pages 796--809. SIAM, 2021.
\newblock \href {https://doi.org/10.1137/1.9781611976465.50}
  {\path{doi:10.1137/1.9781611976465.50}}.

\bibitem[CCHM15]{parameterized-streaming}
Rajesh Chitnis, Graham Cormode, Mohammad~Taghi Hajiaghayi, and Morteza
  Monemizadeh.
\newblock Parameterized streaming: Maximal matching and vertex cover.
\newblock In {\em Proceedings of the 2015 Annual ACM-SIAM Symposium on Discrete
  Algorithms (SODA)}, pages 1234--1251, 2015.
\newblock \href {https://doi.org/10.1137/1.9781611973730.82}
  {\path{doi:10.1137/1.9781611973730.82}}.

\bibitem[CFK{\etalchar{+}}15]{parameterized-algorithms-book}
Marek Cygan, Fedor~V. Fomin, Lukasz Kowalik, Daniel Lokshtanov, Daniel Marx,
  Marcin Pilipczuk, Michal Pilipczuk, and Saket Saurabh.
\newblock {\em Parameterized Algorithms}.
\newblock Springer Publishing Company, Incorporated, 1st edition, 2015.
\newblock \href {https://doi.org/10.5555/2815661} {\path{doi:10.5555/2815661}}.

\bibitem[DKT14]{dynamic-tree-depth-decomposition}
Zden{\v{e}}k Dvo{\v{r}}{\'a}k, Martin Kupec, and Vojt{\v{e}}ch T{\r{u}}ma.
\newblock A dynamic data structure for mso properties in graphs with bounded
  tree-depth.
\newblock In Andreas~S. Schulz and Dorothea Wagner, editors, {\em Algorithms -
  ESA 2014}, pages 334--345, Berlin, Heidelberg, 2014. Springer Berlin
  Heidelberg.
\newblock \href {https://doi.org/10.1007/978-3-662-44777-2_28}
  {\path{doi:10.1007/978-3-662-44777-2_28}}.

\bibitem[DL78]{schwartz-zippel-1}
Richard~A. Demillo and Richard~J. Lipton.
\newblock A probabilistic remark on algebraic program testing.
\newblock {\em Information Processing Letters}, 7(4):193--195, 1978.
\newblock \href {https://doi.org/10.1016/0020-0190(78)90067-4}
  {\path{doi:10.1016/0020-0190(78)90067-4}}.

\bibitem[DT13]{dynamic-parameterized-counting}
Zden\v{e}k Dvo\v{r}\'{a}k and Vojt\v{e}ch T\r{u}ma.
\newblock A dynamic data structure for counting subgraphs in sparse graphs.
\newblock In {\em Proceedings of the 13th International Conference on
  Algorithms and Data Structures}, WADS'13, page 304–315, Berlin, Heidelberg,
  2013. Springer-Verlag.
\newblock \href {https://doi.org/10.1007/978-3-642-40104-6_27}
  {\path{doi:10.1007/978-3-642-40104-6_27}}.

\bibitem[EG59]{edge-bound}
P.~Erdős and T.~Gallai.
\newblock On maximal paths and circuits of graphs.
\newblock {\em Acta Mathematica Academiae Scientiarum Hungarica}, 10:337--356,
  1959.
\newblock \href {https://doi.org/10.1007/BF02024498}
  {\path{doi:10.1007/BF02024498}}.

\bibitem[HHS21]{dynamic-algs-survey}
Kathrin Hanauer, Monika Henzinger, and Christian Schulz.
\newblock Recent advances in fully dynamic graph algorithms.
\newblock {\em CoRR}, abs/2102.11169, 2021.
\newblock URL: \url{https://arxiv.org/abs/2102.11169}, \href
  {http://arxiv.org/abs/2102.11169} {\path{arXiv:2102.11169}}.

\bibitem[HKNS15]{HKNS-lbs}
Monika Henzinger, Sebastian Krinninger, Danupon Nanongkai, and Thatchaphol
  Saranurak.
\newblock Unifying and strengthening hardness for dynamic problems via the
  online matrix-vector multiplication conjecture.
\newblock In {\em Proceedings of the forty-seventh annual ACM symposium on
  Theory of computing}, pages 21--30, 2015.

\bibitem[HLNW17]{so-hardness}
Monika Henzinger, Andrea Lincoln, Stefan Neumann, and Virginia~Vassilevska
  Williams.
\newblock Conditional hardness for sensitivity problems.
\newblock {\em CoRR}, 2017.
\newblock URL: \url{http://arxiv.org/abs/1703.01638}, \href
  {http://arxiv.org/abs/1703.01638} {\path{arXiv:1703.01638}}.

\bibitem[IO14]{Iwata2014FastDG}
Yoichi Iwata and Keigo Oka.
\newblock Fast dynamic graph algorithms for parameterized problems.
\newblock {\em ArXiv}, abs/1404.7307, 2014.
\newblock \href {http://arxiv.org/abs/1404.7307} {\path{arXiv:1404.7307}}.

\bibitem[Kou08]{koutis-packing}
Ioannis Koutis.
\newblock Faster algebraic algorithms for path and packing problems.
\newblock In {\em Automata, Languages and Programming}, pages 575--586, Berlin,
  Heidelberg, 2008.
\newblock \href {https://doi.org/10.1007/978-3-540-70575-8_47}
  {\path{doi:10.1007/978-3-540-70575-8_47}}.

\bibitem[Kou12]{koutis-constraints}
Ioannis Koutis.
\newblock Constrained multilinear detection for faster functional motif
  discovery.
\newblock {\em Information Processing Letters}, 112(22):889--892, November
  2012.
\newblock \href {https://doi.org/10.1016/j.ipl.2012.08.008}
  {\path{doi:10.1016/j.ipl.2012.08.008}}.

\bibitem[KW16]{algebraic-problems}
Ioannis Koutis and Ryan Williams.
\newblock Limits and applications of group algebras for parameterized problems.
\newblock {\em ACM Trans. Algorithms}, 12(3), 5 2016.
\newblock \href {https://doi.org/10.1145/2885499} {\path{doi:10.1145/2885499}}.

\bibitem[Pat10]{patrascu-lbs}
Mihai Patrascu.
\newblock Towards polynomial lower bounds for dynamic problems.
\newblock In {\em Proceedings of the forty-second ACM symposium on Theory of
  computing}, pages 603--610, 2010.

\bibitem[Sch80]{schwartz-zippel-2}
J.~T. Schwartz.
\newblock Fast probabilistic algorithms for verification of polynomial
  identities.
\newblock {\em J. ACM}, 27(4):701–717, 10 1980.
\newblock \href {https://doi.org/10.1145/322217.322225}
  {\path{doi:10.1145/322217.322225}}.

\bibitem[Tsu19]{best-kpath-directed}
Dekel Tsur.
\newblock Faster deterministic parameterized algorithm for k-path.
\newblock {\em Theoretical Computer Science}, 790:96--104, 2019.
\newblock \href {https://doi.org/https://doi.org/10.1016/j.tcs.2019.04.024}
  {\path{doi:https://doi.org/10.1016/j.tcs.2019.04.024}}.

\bibitem[vdBS19]{algebraic-distance-oracles}
Jan van~den Brand and Thatchaphol Saranurak.
\newblock Sensitive distance and reachability oracles for large batch updates.
\newblock In {\em 2019 IEEE 60th Annual Symposium on Foundations of Computer
  Science (FOCS)}, pages 424--435, 2019.
\newblock \href {https://doi.org/10.1109/FOCS.2019.00034}
  {\path{doi:10.1109/FOCS.2019.00034}}.

\bibitem[Wil09]{williams-k-path}
Ryan Williams.
\newblock Finding paths of length k in $\mathcal{O}^*(2^k)$ time.
\newblock {\em Inf. Process. Lett.}, 109(6):315–318, 2 2009.
\newblock \href {https://doi.org/10.1016/j.ipl.2008.11.004}
  {\path{doi:10.1016/j.ipl.2008.11.004}}.

\bibitem[W\l19]{extensor-mult}
Micha\l{} W\l{}odarczyk.
\newblock Clifford algebras meet tree decompositions.
\newblock {\em Algorithmica}, 81(2):497–518, 2 2019.
\newblock \href {https://doi.org/10.1007/s00453-018-0489-3}
  {\path{doi:10.1007/s00453-018-0489-3}}.

\bibitem[Zip79]{schwartz-zippel-3}
Richard Zippel.
\newblock Probabilistic algorithms for sparse polynomials.
\newblock In {\em Proceedings of the International Symposiumon on Symbolic and
  Algebraic Computation}, EUROSAM '79, page 216–226. Springer-Verlag, 1979.
\newblock \href {https://doi.org/10.5555/646670.698972}
  {\path{doi:10.5555/646670.698972}}.

\end{thebibliography}

\end{document}